\documentclass[authoryear,12pt, a4paper]{article}
\usepackage{caption}
\usepackage{amsmath, mathpazo, amssymb, amsfonts, amsthm}
\usepackage{algorithm2e}
\usepackage[T1]{fontenc}
\usepackage[colorlinks=true,linkcolor=red, citecolor=blue, urlcolor=blue]{hyperref}
\usepackage[utf8]{inputenx}
\usepackage{graphicx}
\usepackage{multirow}
\usepackage{dcolumn}
\usepackage{pdflscape}
\usepackage{natbib}
\usepackage{hyperref}
\newcolumntype{d}[1]{D{.}{.}{-1}}
\newcolumntype{s}[1]{D{,}{,}{-1}}
\usepackage[left=2cm,top=3cm,right=2cm,bottom=3cm,includefoot]{geometry}
\usepackage{ulem}
\usepackage{enumitem}
\usepackage{xspace}
\usepackage[group-separator={,},group-minimum-digits=4]{siunitx}
\usepackage[hang,flushmargin]{footmisc}
\usepackage[toc, title]{appendix}
\usepackage{setspace}
\usepackage{bookmark}
\bookmarksetup{open,numbered}

\newcounter{myenumi}
\newenvironment{myenumerate}{
\begin{enumerate} 
\setcounter{enumi}{\value{myenumi}}
}{
\setcounter{myenumi}{\value{enumi}}
\end{enumerate}
}

\theoremstyle{plain}
\newtheorem{theorem}{Theorem}
\newtheorem{lemma}{Lemma}
\newtheorem{proposition}{Proposition}

\theoremstyle{definition}
\newtheorem{definition}{Definition}

\newtheorem*{remark*}{Remark}
\newtheorem{example}{Example}
\newtheorem{corollary}{Corollary}

\DeclareMathOperator{\Ch}{{Ch}}
\DeclareMathOperator{\E}{\mathbb{E}}

\newcommand{\SPDA}{SP-DA\xspace}
\newcommand{\SRDA}{SR-DA\xspace}

\usepackage{array}
\newcommand{\PreserveBackslash}[1]{\let\temp=\\#1\let\\=\temp}
\newcolumntype{C}[1]{>{\PreserveBackslash\centering}p{#1}}
\newcolumntype{R}[1]{>{\PreserveBackslash\raggedleft}p{#1}}
\newcolumntype{L}[1]{>{\PreserveBackslash\raggedright}p{#1}}

\usepackage{clipboard}
\newclipboard{empirical-large-core}

\usepackage[capitalize,noabbrev]{cleveref}
\usepackage{appendix}
\usepackage{etoolbox}
\AtBeginEnvironment{appendices}{\crefalias{section}{appendix}}
\newif\ifinappendix\inappendixfalse
\crefname{apptab}
  {\protect{\ifinappendix\else Appendix \fi}Table}
  {\protect{\ifinappendix\else Appendix \fi}Table}
\AtBeginEnvironment{appendices}{\inappendixtrue \crefalias{table}{apptab}}

\usepackage{titlesec}
\titlespacing*{\section}
{0pt}{2.5ex plus 1ex minus .2ex}{2ex plus .2ex}
\titlespacing*{\subsection}
{0pt}{2.25ex plus 1ex minus .2ex}{1.2ex plus .2ex}
\titlespacing*{\paragraph}
{0pt}{2.25ex plus 1ex minus .2ex}{1em}

\makeatletter 
\renewcommand\paragraph{\@startsection{paragraph}{4}{\z@}%
                                    {1.7ex \@plus1ex \@minus.2ex}%
                                    {-1em}%
                                    {\normalfont\normalsize\bfseries}}
\makeatother

\title{The Large Core of College Admission Markets: Theory and Evidence\footnote{This version: August, 2022. We thank Yannai Gonczarowski, Jacob Leshno, Bobby Pakzad-Hurson, Deborah Marciano, Thayer Morrill, Alex Rees-Jones, Al Roth, Ben Roth, Jan Christoph Schlegel, Alex Teytelboym and seminars and conferences participants for helpful discussions and suggestions. We are especially grateful to Scott Kominers, who discussed an earlier version of this paper in the 2018 ASSA meeting, for many useful suggestions. This research was supported by a grant from the United States -- Israel Binational Science Foundation (BSF), Jerusalem, Israel. This research was support by the Israel Science Foundation (grants no. 1241/12 and 1780/16). Assaf Romm is supported by a Falk Institute grant. The Hungarian Higher Education Application Database (FELVI) is owned by the Hungarian Education Bureau (Oktatasi Hivatal). The data were processed by the Databank of KRTK (Centre for Economic and Regional Studies).}}

\author{Péter Biró \and Avinatan Hassidim \and Assaf Romm \and Ran I. Shorrer \and Sándor Sóvágó\thanks{Biró: Center for Economic and Regional Studies--Economics Institute, Corvinus University of Budapest (biro.peter@krtk.mta.hu); Hassidim: Bar Ilan University (avinatan@macs.biu.ac.il); Romm: Hebrew University of Jerusalem (assafr@gmail.com); Shorrer: Pennsylvania State University (rshorrer@gmail.com); Sóvágó: University of Groningen (s.sovago@rug.nl).}}

\date{}
\begin{document}

\pagenumbering{gobble}
\maketitle

\begin{abstract}
We study stable allocations in college admissions markets where students can attend the same college under different financial terms. The deferred acceptance algorithm identifies a stable allocation where funding is allocated based on merit. While merit-based stable allocations assign the same students to college, non-merit-based stable allocations may differ in the number of students assigned to college. In large markets, this possibility requires heterogeneity in applicants' sensitivity to financial terms. In Hungary, where such heterogeneity is present, a non-merit-based stable allocation would increase the number of assigned applicants by 1.9\%, and affect 8.3\% of the applicants relative to any merit-based stable allocation. These findings contrast sharply with findings from the matching (without contracts) literature.
\end{abstract}
\clearpage

\pagenumbering{arabic}

\section{Introduction}

\Copy{limitat}{In} recent years, a growing number of students are being assigned to schools through centralized clearinghouses. The success of such clearinghouses crucially relies on the use of a stable matching mechanism \citep{rothxing1994, roth2002}.\footnote{Stability is also useful for predicting behavior in decentralized matching markets \citep[e.g.,][]{banerjee2013}.}
The matching market design literature finds that a designer who wishes to implement a stable allocation has limited scope for further design. First, the rural hospital theorem determines that the same positions are filled in all stable allocations \citep{roth1984,roth1986}. Second, the set of stable allocations has the consensus property: all students prefer the outcome of the student-proposing deferred acceptance mechanism (henceforth \SPDA) to any other stable allocation \citep{gs1962,roth1984contracts}. Third, empirical and theoretical studies suggest that all students, save for a handful, receive the same assignment in all stable allocations \citep[e.g.,][]{akl2017, al2016, im2005x, kp2009, rp1999}. This last finding implies that schools have limited incentive to collect information and to misreport their preferences \citep{demange1987}.

The above-mentioned results apply to two-sided matching markets (men and women, students and schools, etc.) where agents' preferences are over potential partners from the other side. However, the environments studied and designed by economists are often more complex. For example, college applicants care not only about the study program they are assigned to, but also about the level of financial aid they receive. In this paper, we ask whether the set of stable allocations continues to be small in these more complex environments.


We study Hungarian college admissions, where colleges offer multiple levels of financial aid. In this environment, we show theoretically and empirically that it is no longer true that almost all students receive the same assignment in all stable allocations. Furthermore, different stable allocations differ in the set of positions filled and in the composition of assigned students, and there is no consensus among students on the most-preferred stable allocation. 


\medskip To gain intuition, we provide a simple example. The example is a minimal instance of our model of Hungarian college admissions.\footnote{Our model captures the structure of preferences observed in other centralized college admissions markets, including Turkey, Australia, Israel, Ukraine, Russia, and the US.} We accompany the verbal description with the notation of our model, which we introduce in \cref{sec:theory}. 

\begin{example} \label{DifferenrNumber}
There are two students, $S=\{r,p\}$, and one college, $C=\{c\}$. The college has two seats, but only one of these seats is state-funded (formally, $q^0_c=1$, $q^1_c=1$).  The college finds the rich student, $r$, more attractive than the poor student, $p$ (formally, $r\gg_c p$). The college prefers to accept the most attractive students, and to fill its capacity. The rest of the college's preferences over acceptable allocations are fully described by
\begin{align*}
\left\{(r,c,0),(p,c,1)\right\} & \succ_c \left\{(r,c,1),(p,c,0)\right\} \succ_c \left\{(r,c,0)\right\} \succ_c \\
\left\{(r,c,1)\right\} & \succ_c \left\{(p,c,0)\right\} \succ_c \{(p,c,1)\},
\end{align*}
where 1 (0) indicates admission with (without) state funding. 

The rich student, $r$, prefers to receive state funding, but she is willing to attend the college even if she does not get state funding. Formally, $r$'s preferences are $(r,c,1) \succ_r (r,c,0)\succ_r\emptyset$. By contrast, the poor student, $p$, is only interested in admission with state funding. Thus, $p$'s preferences are summarized by $(p,c,1)\succ_p\emptyset$. 

There are two stable allocations in this market: the result of \SPDA, $\left\{(r,c,1)\right\}$, where funding is allocated based on merit; and $\left\{(r,c,0), (p,c,1)\right\}$, where more positions are filled and funding is not allocated based on merit. We note that the only allocation that both students weakly prefer to both of these stable allocations has both students receiving state funding, but this allocation is not acceptable to the college. Therefore, there is no consensus among students on a student-optimal stable allocation.
\end{example}

Can we expect stable allocations to be meaningfully different in realistically-sized markets? To answer this question, we develop a large market model of Hungarian college admissions markets and study the assumptions required from this model to guarantee that the set of stable allocations is large. Our main theoretical result (\cref{prop:big}) is that the set of stable allocations is large when students are heterogeneous in their sensitivity to financial terms (like the rich student and the poor student in \cref{DifferenrNumber}). By contrast, when students are not heterogeneous in their sensitivity to financial terms---that is, when all students find both contracts with the same college to be nearly perfect substitutes or when all students do not find them to be close substitutes---then the set of stable allocations is small. We also show that both variants of DA allocate funding based on merit (\cref{prop:merit}) and that the set of merit-based allocations is small. 
 
The proof of \cref{prop:big} as well as our empirical analysis rely on the {\it preference flip algorithm}, a novel stable algorithm. To gain intuition on how our algorithm works, we revisit \cref{DifferenrNumber}. If each type of seat was controlled separately, then competition between the seats would rule out the stable allocation in which both students are assigned. However, since one college controls both seats, it can make sure that they do not compete with each other. Given a stable allocation, colleges have \textit{local market power} over students who are admitted with state funding, but who have no outside option (a contract at another college or being unassigned) that they prefer to the self-funded contract with the same college. An extreme case is when students rank the state-funded and self-funded contracts with the same college consecutively, like the rich student in our example. In this case, colleges can exercise market power by refusing to accept such students to state-funded seats, thus freeing up the state-funded seats which can then be used to recruit price-sensitive students like the poor student in our example.\footnote{This argument is not precise, since by recruiting new students the college may change the outside options available to its other students. Our analysis takes this challenge into account.} The preference flip algorithm can be interpreted as allowing colleges to exercise their local market power.

We use the preference flip algorithm for exploring the size of the set of stable allocations. However, implementing it will require more than merely changing the matching algorithm. Specifically, the preference flip algorithm requires colleges to offer financial aid selectively based on students' sensitivity to financial terms, giving students an incentive to misrepresent their sensitivity to financial terms. Still, colleges can use hard information to infer students' sensitivity to financial terms. For example, in the American context, such information is available to colleges through the Free Application for Federal Student Aid (FAFSA) form. In addition, many American colleges hire companies to personalize their financial aid packages using large propitiatory datasets \citep{lieber2021price,CP2019}.

We next assess the size of the set of stable allocations empirically using administrative data from the Hungarian college admissions system. We document substantial heterogeneity in applicants' sensitivity to financial terms. In line with our large market analysis, both variants of DA result in essentially identical allocations, while the outcome of the preference flip algorithm differs for more than 9,000 applicants (approximately 8 percent of the applicants), with slightly more winners than losers. More importantly, the preference flip algorithm increases the number of applicants assigned to college by 1.9 percent (approximately 1,500) with approximately 2,100 unassigned applicants gaining admission and approximately 600 losing their place.\footnote{Having found that the set of stable allocations could be large and in the absence of the consensus property, a natural objective is to identify stable allocations that are optimal according to some criterion. However, we show that achieving this objective may not be feasible with current computing technology. Specifically, \cref{prop:np} shows that finding the largest stable allocation is NP-hard. This means that finding an efficient (polynomial-time) algorithm for solving this problem will establish that $P=NP$ (which is widely believed to be false).} Since programs in the capital---which are generally prestigious and highly demanded---typically fill all their positions both under DA and under our alternative, the gains in enrollment mostly accrue to colleges in the periphery.
 
We analyze the characteristics of those who benefit and lose from the preference flip algorithm. Applicants who prefer the outcome of the non-merit-based preference flip algorithm to the outcome of DA come from lower socioeconomic background relative to applicants who prefer the outcome of DA. The preference flip algorithm also increases geographic mobility by increasing the number of applicants assigned to a college outside their county of residence.

Our data does not include hard information that would allow inferring applicants' sensitivity to financial terms. For this reason, our main analysis uses an algorithm which only applies local market power over students who are extremely insensitive to financial aid (as revealed by their preference reports). Had hard information been collected, one can argue that these students could plausibly be identified. In \cref{app:narrow_and_broad}, we expand the analysis in two directions. First, we assess the consequences of being able to detect only a subset of these students. Second, we analyze algorithms that apply market power more broadly. We find that the effectiveness of our approach increases with the quality of information on students' sensitivity to funding.

Our findings speak to the broader research question of how to design college admissions processes in the presence of constraints on additional resources such as financial aid, dormitories, etc. Our analysis does not consider many policy levers that are available to policymakers. For example, we fix the number of scholarships the government allocates to each program, we keep scholarships indivisible, and we require that admission and financial terms be determined simultaneously based on preference reports and that the outcome is stable. Even under these restrictions, we find a substantial scope for market design to affect the outcome.

The remainder of the paper is organized as follows. Following a short review of related literature, \cref{sec:institutional} describes college admissions in Hungary. \cref{sec:theory} introduces a formal model of Hungarian college admissions. \cref{sec:theory1} explores the set of stable allocations. \cref{sec:large_market} includes the large market analysis. \cref{sec:results} reviews the empirical findings. \cref{sec:discussion} concludes. Proofs are relegated to the appendix throughout.

\subsection{Related Literature}

Empirical and theoretical studies show that the set of stable allocations is typically small \citep{akl2017, al2016, im2005x, kp2009, rp1999, storms2013}. Prior to these studies, a large portion of the literature on the theory of two-sided matching markets was motivated by the potential multiplicity of stable allocations. Examples include studies of the structure of the set of stable allocations \citep{knuth1976}, of fair stable allocations  \citep{kk2006, ys2011}, and of incentives \citep{em2007, roth1982, sonmez1999}. The rural hospital theorem \citep{roth1984,roth1986} refuted suggestions that changing the way the National Residency Match Program treats medical graduates and hospitals may change the number of doctors assigned to rural hospitals.

Truthful reporting to the student-proposing DA mechanism is a weakly dominant strategy for students, and there is no stable matching mechanism that makes truthful reporting dominant for both sides of the market \citep{df1981, roth1982}. \citet{demange1987} show that schools' incentives to manipulate \SPDA are intimately related to the multiplicity of stable allocations.\footnote{Numerous studies have analyzed the optimal behavior of schools when the \SPDA mechanism is in place \citep[e.g.,][]{cs2014, ehlers2004,ku2006, rr1999,sonmez1997}.} Several studies show that in large markets it is safe for schools to report their true preferences to \SPDA \citep[][]{akl2017, ab2012, im2005x, kp2009, lee2017}.

Complex two-sided matching markets are studied in the matching-with-contracts literature \citep{hm2005}. Much of this literature focuses on identifying conditions under which \SPDA remains stable and strategy-proof for students. Examples include \citet{Fleiner2003}, \citet{hk2010b}, \citet{hm2005}, \citet{hkw2015}, \citet{kc1982}, and \citet{roth1984contracts}.\footnote{Applications related to college admissions include \citet{abizada2016}, \citet{af2017}, \citet{aygunbo}, \citet{orhan_turhan_2017}, \citet{nei2016strategic}, \citet{pakzad2014stable}, \citet{westkamp2013}, and \citet{yenmez2015}.} Most closely related is \citet{hk2015}. 

\citet{hassidim2017need} use the results of \citet{hk2015} to show that \SPDA is stable and strategy-proof in more general college admissions environments. They show that different stable allocations may have different cardinality. The proofs in our paper rely on a different construction, which allows us to derive further results. We contribute to \citet{hk2015} and \citet{hkw2015} in several ways. First, we identify a real-life market design application where colleges' choice functions satisfy the hidden substitutes condition but fall outside the (unilateral) substitutes domain.\footnote{To see that college preferences are not (unilaterally) substitutable \citep{hk2010b}, note that in \cref{DifferenrNumber}, college $c$ rejects $(p, c, 1)$ from the menu $\left\{(p, c, 1), (r, c, 1)\right\}$ but does not reject this contract from the larger menu $\left\{(p, c, 1),(r, c, 1),(r, c, 0)\right\}$. 
Colleges' preferences also do not meet any other condition that guarantees the existence of a student-optimal stable allocation \citep{hrs2019}.} Second, we provide a large market analysis. Third, we provide a complete characterization of the set of stable allocations in the domain of Hungarian college admissions markets.\footnote{The substitutable completions approach of \citeauthor{hk2015} guarantees existence, but it is not guaranteed to identify all stable allocations. Indeed, one can verify that, in \cref{DifferenrNumber}, for any substitutable completion of the college's choice function, $Y^{\SPDA}$ is the unique allocation that is stable with respect to the completion.} Fourth, we show that different stable allocations may have different cardinality in this more restrictive domain, and that stable allocations in this domain may result in fewer assigned students relative to \SPDA. Fifth, we consider the practical aspect of computation.

Large market models have been studied extensively in the market design literature. Examples include \citet{an}, \citet{abh2014, anr, akl2017}, \citet{al2016}, \citet{immorlica2020information}, \citet{kp2009}, \citet{kpr2013}, \citet{lee2017}, and \citet{leshno2017cutoff}. Most closely related are \citet{im2005x, im2015} and \citet{rheinganslarge} who show that even in the absence of contracts, certain regularities in preferences can lead to a large set of stable allocations. These papers study settings where the rural hospital theorem and the consensus property apply. By contrast, we study markets with contracts and show that different stable allocations may differ in the number of students assigned to college and where the consensus property does not hold.

Our paper is related to studies of reserve design \citep{orhan_turhan_2017,dkps2018,AYGUN2020105069,H1b}. In the context of school choice, a key observation in this literature is that students may be indifferent between different seats in the same school, but some seats are reserved for certain groups of students. This implies, using our terminology, that schools have local market power over students from these groups. The reserve-design literature focuses on the effect of different ways the mechanism can break students' preference ties in order to form strict student rankings of contracts, while keeping the priorities at each seat fixed. By contrast, we study an environment where students have strict preferences over all contracts.

Our results shed light on the policy debate on market power in higher education \citep[see, e.g.,][]{hoxby2000}. This literature gained traction after the U.S. Department of Justice brought an antitrust case against a group of elite colleges  for sharing prospective students' financial information and coordinating their financial aid policy. MIT contested the charges, claiming that this practice prevents bidding wars over the best students and thus frees up funds to support needy students, and that MIT does not  profit financially from this practice \citep{overlap}. In 1994, Congress passed the Improving America's Schools Act, whose Section~568 permits some coordination and the sharing of information between  institutions with a need-blind admissions policy. Our findings provide support to MIT's arguments. We show that even in the absence of a motive to increase profit, colleges have an incentive to apply market power in order to improve the quality of their incoming cohorts, and that the consequences for students are heterogeneous. On average, needy students gain, in part at the expense of more wealthy students. Furthermore, information about students' sensitivity to financial aid, is necessary to facilitate such behavior.

\section{Background: College Admissions in Hungary}\label{sec:institutional}

Each year about 80,000 applicants are assigned to undergraduate college programs in Hungary.
Applicants are assigned to \textit{programs}, i.e., a specific major in a specific college.\footnote{We focus on admissions to undergraduate programs, which include three types of programs: bachelor's degree programs (with a typical study duration of 3 years), combined bachelor--master's programs (with a typical study duration of 5 years), and tertiary vocational education (with a typical study duration of 2 years).} The admissions process has been centralized since 1985. The centralized clearinghouse has been using variants of DA to assign applicants to colleges since 1997.  

Historically, higher education in Hungary was free. However, in recent years, the government has capped the number of state-funded (free) seats, and programs are allowed to offer admission to self-funded (tuition-paying) seats as well.\footnote{Citizens of the European Economic Area who have not yet graduated from higher education are eligible for state funding.} To implement this change, the centralized clearinghouse started requesting applicants to submit a rank-order list (ROL) of \textit{alternatives}, specifying that the student attends a program under particular financial terms (e.g., economics in the University of Debrecen with state funding). This flexibility was crucial in order to allow applicants to express their preferences. For example, some applicants may not be willing or able to afford paying tuition, while others may have a strong preference for certain programs, and be less price-sensitive.

Applicants may rank as many alternatives as they wish. Submitting an ROL with three programs (corresponding to up to six alternatives) only requires paying the fixed application fee of \$50 (9,000HUF). Applicants are required to pay a registration fee of \$11 (2,000HUF) for each additional program in their ROL. Disadvantaged applicants are exempt from these fees.\footnote{To be eligible for disadvantaged status, an applicant must have a per-capita household income that is lower than 130 percent of the minimum pension (i.e., lower than approximately \$1,500 a year).}

College programs report to the mechanism their capacity (i.e., the maximal number of applicants they can accept). Admissions priorities in each program are based on a weighted average of several variables (mainly academic performance in the 11th and 12th grades and matriculation exam scores, but also credits for disadvantaged and disabled applicants, for applicants who demonstrate fluency in another language, and for a small number of gifted applicants). Different study programs may use different weighting schemes (e.g., a computer science program may assign a greater weight to physics grades relative to a psychology program).

In 2007, the year we focus on, to accommodate the multiple modes of financing, the mechanism required that the programs also report the number of state-funded seats they offered. This took place prior to the beginning of the admissions process.\footnote{In 2008, Hungarian college admissions adopted a more complex variant of DA \citep{BIRO20103136, shorrersovago2018}. Although our analysis continues to hold under the new mechanism, we focus on 2007 to highlight that our results do not rely on the new, unique features.} The mechanism then created two ``auxiliary programs'' with capacities corresponding to the number of state-funded and self-funded seats in the program, and endowed both with the priorities of the program.\footnote{This description applies to full-time programs. There was an additional constraint on the total number of state-funded seats in each field of study, which was only binding for part-time programs. Additionally, part-time programs in computer science and engineering use slightly different priorities for state-funded and self-funded seats (they assign a lower weight to mathematics scores of applicants for self-funded seats relative to state-funded seats). Our empirical analysis takes this into account (see \cref{app:section:empirical_findings}).} Finally, student-receiving DA (\SRDA)---i.e., auxiliary-program-proposing DA---was used to determine the final assignment.

\Cref{tab:alternatives} presents information on the alternatives that colleges offered in 2007. The applicants' choice set consisted of 3,740 alternatives, corresponding to 2,289 study programs. The number of study programs that offered both state-funded and self-funded alternatives was 1,451. Some study programs were available with state-funding exclusively (172), and other programs---mostly part-time programs---were available with self-funding exclusively (666). Panel B shows that  93,999 of the 111,685 available seats (84.1 percent) were available in study programs that offered both state-funded and self-funded alternatives. The overwhelming majority of the seats that were offered with self-funding exclusively (about 91 percent) were in the form of part-time education.

\begin{table}[htpb!]
\centering\footnotesize
\caption{Applicants' choice set}
\vspace{-1.5em}
\label{tab:alternatives}
\begin{tabular}{l S[table-number-alignment = center, table-format = 5.0] S[table-number-alignment = center, table-format = 5.0] S[table-number-alignment = center, table-format = 5.0]} \\ \hline\hline
 & \multicolumn{1}{c}{Total} & \multicolumn{1}{c}{Full-time} & \multicolumn{1}{c}{Part-time} \\  & \multicolumn{1}{c}{(1)} & \multicolumn{1}{c}{(2)} & \multicolumn{1}{c}{(3)} \\ \hline \multicolumn{4}{l}{A. {\it Alternatives}} \\
Number of programs &         2289 &         1357 &          932 \\
\multicolumn{1}{l}{-- state- and self-funded} &         1451 &         1133 &          318 \\
-- state-funded exclusively &          172 &          158  &           14 \\
-- self-funded exclusively &          666 &           66  &          600 \\
Number of alternatives &         3740 &         2490 &         1250 \\
& & & \\
\multicolumn{4}{l}{B. {\it Capacities}} \\
\multicolumn{1}{l}{State-funded and self-funded} &        93999 &        75197 &         18802 \\
 -- state-funded &        47809 &        43084 &          4725 \\
 -- self-funded &        46190 &        32113 &         14077 \\
State-funded exclusively &          917 &          806 &           111 \\
Self-funded exclusively &        16769 &          691 &         16078 \\ \hline 
Total &       111685 &        76694 &         34991 \\ \hline\hline 
\multicolumn{4}{p{9.5cm}}{{\it Notes}: The table presents summary statistics of the applicants' choice set. Panel A presents the number of alternatives by the financial terms. Panel B displays the number of seats by the financial terms.}
\end{tabular}

\end{table}

\paragraph{Timeline.} The application process proceeds as follows. First, the centralized clearinghouse publishes a booklet that includes the rules of the college admissions process together with the list of the alternatives (program--financial terms pairs), and the number of seats available in each alternative in November. Applicants submit their ROLs in mid-February. In May and June, applicants take their matriculation exam. Finally, in mid-July, the clearinghouse notifies applicants about their placement and publishes the priority-score cutoffs for each alternative, i.e., the minimum priority score that was needed to gain admission.

\paragraph{Financing Colleges.} The government funds colleges on a per-student basis, irrespective of the financial terms. Additionally, colleges collect tuition from self-funded students,\footnote{Tuition for a 3-year  bachelor's degree program ranges from \$3,280 to \$6,560  (600,000--1,200,000HUF). } and receive additional compensation for state-funded students from the government. Thus, holding the number of state-funded (self-funded) students constant, colleges' revenue increases with each additional admitted self-funded (state-funded) student.

\section{A Model of Hungarian College Admissions}\label{sec:theory}

We use the many-to-one matching-with-contracts model of \citet{hm2005} to describe \textit{Hungarian college admissions markets}. There is a finite set of \textit{colleges}, $C$, a finite set of \textit{students}, $S$, and a set of \textit{financial terms}, $T=\left\{0,1 \right\}$. A \textit{contract} is a tuple $(s,c,t)\in S\times C \times T$ that specifies a student, a college, and 
financial terms (with $t=1$ representing state funding, and $t=0$ representing self funding). 

An \textit{allocation} is a subset $Y \subseteq S\times C \times T$. An allocation $Y$ is \textit{feasible} if no student is included in more than one contract. Formally, for each student $s$, $\left|Y\cap \left(\left\{ s\right\}\times C\times T \right)\right|\leq 1$.
Given an allocation, $Y$, we let $Y_S$  denote the set of students  involved in some contract in $Y$. Formally, $Y_S :=\left\{s\in S\mid  Y \cap  \left(\left\{ s\right\}\times C\times T \right) \ne \emptyset \right\}$. Similarly, $Y_C := \left\{c\in C\mid  Y \cap  \left( S\times\left\{ c\right\}\times T \right) \ne \emptyset \right\}$. 

\paragraph{Students' preferences} Each student, $s$, has strict preferences over contracts in $\{s\} \times C \times T$  (not all of them are necessarily available) and an outside option which we denote by $\emptyset$. We denote student $s$'s preference relation by $\succ_s$.  Students' preferences, therefore, induce  weak preferences over all feasible allocations, where students only consider their own assignment.  

\paragraph{Colleges' preferences} Each college, $c$, has strict preferences, $\succ_c$, over all feasible allocations in $S\times \{c\} \times T$.\footnote{None of our results relies on the assumption that colleges' preferences are strict.} These preferences induce weak preferences over all feasible allocations, where $c$ strictly prefers $Y$ to $Y'$ if and only if $Y\cap \left(S\times \{c\} \times T\right) \succ_c Y'\cap \left(S\times \{c\} \times T\right)$. The preferences of each college, $c$, satisfy the following conditions. 

First, $c$ is associated with two numbers, $q^1_c$ and $q^0_c$, representing a constraint on the number of students that can be accepted under each of the financial terms. The college $c$ prefers the empty allocation to all allocations that \textit{violate $c$'s quotas}, that is, that assign to $c$  more than $q^t_c$ students under the financial terms $t$.  

Second, $c$ has a complete order over $S \cup \{\emptyset\}$, denoted by $\gg_c$, representing a ranking over students. Given an allocation  $Y \subseteq S \times \{c\} \times T$ (an assignment of students to the college under some financial terms),  if $q_c^t$ is not binding  then the college prefers to accept an additional student $s$ under the financial terms $t$ if and only if $s\gg_c \emptyset$. Formally, if $s\notin Y_S$ and $\left|Y\cap \left( S\times \{c\} \times \{t\}\right)\right| < q_c^t$, then 
$Y\cup\left\{\left(s,c,t \right)\right\}\succ_c Y$ iff $s\gg_c\emptyset$. Additionally, $c$ prefers to replace student $s$ who receives the financial terms $t$ with another student $s'$ who is not assigned to $c$ (under the same financial terms) if and only if $s'\gg_c s$. Formally, for all ${Y\subseteq S\times \{c\} \times T}$, if $(s,c,t)\in Y$ and $s'\notin Y_S$, then $ \left(Y\cup \left\{\left(s',c,t \right)\right\}\right)\setminus \left\{\left(s,c,t \right)\right\} \succ_c Y$ iff $s'\gg_c s$.

Third, so long as quotas are not violated, the composition of the incoming cohort is lexicographically more important to the college relative to the way funding is allocated. Formally, let $Y,Y'\subseteq S\times \{c\} \times T $ be two feasible allocations that do not violate $c$'s quotas. Then if $Y_S=Y'_S$ (i.e., $Y$ and $Y'$ differ only in the identity of the recipients of state funding in $c$) and $Y\succ_c Y'' \succ_c Y'$ for some $Y''\subseteq S\times \{c\} \times T$, it follows that $Y''_S=Y_S$.   	

\paragraph{Discussion of theoretical assumptions} We refrain from assuming that students always prefer to receive state funding for several reasons. First, this generality allows our model to capture environments where the ranking of alternatives is more ambiguous, such as admission with and without a dormitory assignment. Second, in practice, some students are not eligible to receive financial aid and a small fraction of applicants report unusual preferences \citep{shorrersovago2018}.

According to our assumptions, subject to quotas,  colleges' preferences depend more on the size and quality of the incoming cohort than they do on the distribution of funding. This assumption is not unusual in the literature \citep[e.g.,][]{kim2010,heo2017financial}, and is consistent with preferences of colleges in other markets \citep[e.g.,][]{hrspp}.  Our assumptions rule out redistributive motives. As highlighted by \citet{CP2019}, the presence of such motives will strengthen our arguments.  

Our model also abstracts from upper quotas on state-funded seats in part-time programs in the same field \citep{BIRO20103136}. We note that all of our theoretical results generalize to this more complex environment. The upper quotas create links between part-time programs, making their choice functions more complex.
The added complexity creates further opportunities to apply local market power (e.g., transfer students between part-time programs). We choose not to take advantage of this feature, which is specific to the Hungarian market, in our empirical analysis.

The special case of our model where $\min\{q_c^0, q_c^1\}=0$ for all $c$ corresponds to \textit{two-sided many-to-one matching (without contracts) markets with responsive preferences} \citep{roth1985college}. Given such a market, we often refer to college $c$'s only positive quota as \textit{$c$'s quota} and to allocations as \textit{matchings}.

\paragraph{Individual rationality} The preferences of each college $c$ induce a choice function $\Ch_c: 2^{S\times C \times T}\rightarrow 2^{S\times \left\{c\right\} \times T}$, that identifies the feasible subset of $Y\cap{\left(S\times \left\{c\right\} \times T\right)}$ most preferred by $c$. Similarly, for students, $\Ch_s:2^{S\times C\times T} \rightarrow 2^{\{s\} \times C \times T}$ 
chooses the most preferred acceptable contract involving $s$ (if one exists). An allocation $Y$ is \textit{individually rational} if all agents choose all the contracts in which they are involved, that is 
$\bigcup_{c\in C} \Ch_c \left( Y \right)=\bigcup_{s\in S} \Ch_s \left( Y \right)=Y$. Individually rational allocations cannot violate quotas since colleges prefer the empty set to allocations that violate their quotas. Furthermore, they are feasible since students never choose more than one contract ($\left|\Ch_s \left(\cdot\right)\right|\leq 1$ for all $s$).

\paragraph{Stability} 
An allocation $Y$ is \textit{blocked (through $Z$)} if there exists a college, $c$, and a non-empty set $Z\subseteq\left({S\times \left\{c\right\} \times T} \right)\setminus Y$ such that $Z \subseteq \Ch_c(Y\cup Z)$ and $Z\subseteq {\cup}_{s\in S}\Ch_s \left( Y \cup Z \right)$.  An allocation is \textit{stable} if it is individually rational and not blocked. In words, an allocation is stable if no coalition of agents  can achieve a weak improvement on its own. 

The following proposition provides a complete characterization of stable allocations.

\begin{proposition} \label{lem:stable_new}
An allocation $Y$ is stable if and only if all of the following hold:
\begin{enumerate}
\item $Y$ is individually rational.
\item $Y$ is not blocked through a singleton $\left\{\left(s,c,t\right)\right\}$. 
\item $Y$ is not blocked through a change in the financial terms of one student, and a contract with a new student using the freed-up space. Formally, for all  $(s,c,1-t)\notin Y$ and $(s',c,t) \in Y$,  $Y$ is not blocked through $\left\{\left(s',c,1-t\right),(s,c,t) \right\}$.
\item $Y$ is not blocked through a change in financial terms that keeps students in the same college (but under different financial terms). Formally, $Y$ is not blocked through $Z$ such that $(s,c,t)\in Z$ implies $(s',c,1-t)\in Y$. 
\end{enumerate}
\end{proposition}

We note that Conditions~3 and 4 of \cref{lem:stable_new} can only be violated in the presence of multiple contractual terms. Therefore, in two-sided many-to-one matching (without-contracts, i.e., when $|T|=1$) markets with responsive preferences the proposition reduces to the statement that a matching is stable if and only if it is individually rational and not blocked by a student--college pair.

\paragraph{Certainly stable allocations}
Allocation $Y$ is \textit{certainly stable} if Conditions 1--3 of \cref{lem:stable_new} hold, as well as the following condition:

\begin{enumerate}
\item[4'.]  For each $(s,c,t)\in Y$, if $(s,c,1-t)\succ_s (s,c,t)$ then $\left|Y\cap\left(S\times \left\{c\right\}\times \left\{1-t\right\} \right)\right|=q_c^{1-t}$ and $(s',c,1-t)\succ_{s'}(s',c,t)$ for all $(s',c,1-t)\in Y$.
\end{enumerate}
When Conditions~1-3 of \cref{lem:stable_new} hold, Condition 4'  
implies Condition~4 of \cref{lem:stable_new}. Therefore, a certainly stable allocation is also stable.\footnote{The converse is not true because students may want to swap their financial terms in the same college in a stable allocation, as long as the college prefers that they do not.} The notion of certain stability proves useful in our empirical analysis because it only relies on quotas and rankings (i.e., $q_c^0$, $q_c^1$, and $\gg_c$) but not on the full description of $\succ_c$, and because there are fewer conditions to verify (each condition refers to no more than three agents).

\paragraph{Comparing stable allocations} When comparing the stable allocation $Y'$ to the stable allocation $Y$, the set of \textit{winners}, $W(Y',Y)$, consists of all students that prefer their assignment under $Y'$. The set of winners comprises of the following subsets:
\begin{enumerate}
\item The set of \textit{newly assigned} students, $N(Y',Y)$, consists of students who are assigned under $Y'$ but not under $Y$.
\item The set of \textit{funding winners}, $FW(Y',Y)$, consists of students who are assigned to the same college in both allocations, but prefer their financial terms under $Y'$. If all students prefer to be funded, students in $FW(Y', Y)$ receive state-funding under $Y'$, but not under $Y$.
\item  The set of \textit{upwards transfers} $UT(Y',Y)):=W(Y',Y)\setminus \left( N(Y',Y) \cup FW(Y', Y)\right)$ consists of winners who get assigned to a new, preferred college. 
\end{enumerate}

We define the set of {\it losers}, and its subsets, {\it newly unassigned}, {\it funding losers}, and {\it downward transfers}, symmetrically.

\paragraph{Merit-based allocations}

We say that an allocation $Y$ \textit{allocates funding based on merit} if 
$s$ weakly prefers $(s,c,t)$ to $(s,c,t')$
whenever $(s,c,t)$ and $(s',c,t')$ belong to $Y$ 
and  $s\gg_c s'$. 
When  $Y$ allocates funding based on merit, within each college, $c$, higher-ranked students according to $\gg_c$ do not envy the financial terms of lower-ranked students who are also assigned to $c$. 
Specifically, if all students prefer to be funded, the recipients of funding in each college are all ranked higher than all self-funded students. A stable allocation that allocates  funding based on merit is a \textit{merit-based stable allocation}.\footnote{
\citet{as2003} establish a connection between stability and justified-envy-freeness in what became the standard school choice framework. 
\citet{rrs2020} study the relation between stability and the elimination of justified envy in a general matching with contracts markets. In the context of multi-dimensional constraints, \citet{delacretaz2016refugee} note that a weak form of no justified envy is independent of stability.}

\paragraph{Related matching markets} The \textit{related (two-sided) matching market (without contracts)} is an auxiliary market where each program-funding-level pair plays the role of an individual college. This mathematical object often simplifies our exposition and analysis. 
Formally, given a Hungarian college admissions market, $\left< S,C,\left\{\succ_c,\gg _c, q_c^0, q_c^1\right\}_{c \in C},\{\succ_s\}_{s \in S} \right>$, the \textit{related  matching market} comprises the set of \textit{auxiliary colleges} ${\hat{C} \equiv C \times T}$ with responsive preferences, where $(c,t)\in \hat{C}$ has quota $q^t_c$ and uses the ranking $\gg_c$, and the set of students $S$ with preferences  $\{\hat{\succ}_s\}_{s \in S}$, such that for every $s\in S$, $c\in C$, and $t\in T$,  $(c,t)\hat{\succ}_s (c',t')$ if and only if $(s,c,t)\succ_s (s,c',t')$. A feasible allocation $Y$ and a matching (in the related matching market) are {\textit{corresponding}} if {$s$ matched with $(c,t)$ if and only if $(s,c,t)\in Y$.}

%

\paragraph{Deferred acceptance}
Student-proposing deferred acceptance takes a profile of preferences as input and outputs an allocation. The allocation is the result of the following process. In each round, each student proposes the most preferred acceptable contract from which she has not yet been rejected, if such a contract exists. Each college then rejects all but the most preferred subset of contracts from these proposals. The algorithm terminates when no proposal is rejected, and the output is the set of contracts proposed in the last round. Given a Hungarian college admissions market, we denote the outcome of \SPDA by $Y^{\text{\SPDA}}$.

Student-receiving deferred acceptance outputs the allocation that corresponds to auxiliary-college-proposing deferred acceptance (formally, the allocation to which this matching corresponds). Of note, in each round of the algorithm, a college may make multiple offers to the same student (one by each auxiliary college). Given a Hungarian college admissions market, we denote the outcome of \SRDA by $Y^{\text{\SRDA}}$.  

In two-sided many-to-one matching (without contracts) markets with responsive preferences, \SRDA is known as college-proposing DA. We intentionally avoid this label to highlight the fact that colleges are not necessarily proposing the most preferred subset of contracts that were not rejected. Since colleges' preferences are not substitutable \citep{hm2005}, such a process would sometimes ``renege'' on proposals that were not rejected due to the rejection of complementary contracts.\footnote{To see that a ``college proposing'' process would sometimes ``renege'' on proposals, consider a modification to \cref{DifferenrNumber} where neither student finds the unfunded seat acceptable. In the first round, the poor student is offered a funded seat and the rich student is offered an unfunded seat, which she declines. As a result, in later rounds, the college would like to renege on the offer made to the poor student, and instead offer the funded seat to the rich student.}

\begin{proposition}\label{lem:outcome_equiv}
$Y^{\text{\SPDA}}$ corresponds to the student-optimal stable matching  in the related market.
\end{proposition}

\begin{corollary}[Strategic properties of DA] \label{prop:dgs}
In Hungarian college admissions markets:
\begin{enumerate}
\item \SPDA is strategy-proof for students.
\item A student can manipulate \SRDA if and only if she strictly prefers $Y^{\text{\SPDA}}$ to $Y^{\text{\SRDA}}$. 
\end{enumerate}
\end{corollary}

\cref{prop:dgs} provides a bound on the number of applicants with incentives to misrepresent their preferences to \SRDA.

\section{The Set of Stable Allocations}\label{sec:theory1}
We begin in \cref{sec:da} by studying the set of merit-based stable allocations. We show that this set is nonempty, and that it includes the allocations that result from both \SPDA and \SRDA. We also show that each college recruits the same number of students across all merit-based stable allocations. In \cref{sec:beyond_da}, we note that different stable allocations differ in the number of assigned students and that there does not necessarily exist a student-optimal stable allocation. We introduce the preference flip algorithm, a (non-merit-based) algorithm designed to increase the number of admitted students compared to \SRDA and prove that it is stable. We then demonstrate that achieving certain policy goals may not be feasible with current computing technology. Specifically, \cref{prop:np} states that finding a largest stable allocation 
is NP-hard.\footnote{Recent advances in the market design literature show that solving some NP-hard problems is feasible in practice \citep[see][]{agoston2016integer, AGOSTON2021, LB_M}.} 

\subsection{Merit-Based Stable Allocations}\label{sec:da}


We begin by establishing a connection between merit-based stable allocations and stable allocations in the related matching market.

\begin{theorem}\label{prop:merit}\label{lem:stable_implies_stable}
Given a Hungarian college admissions market: \begin{enumerate}
\item  $Y$  is a  merit-based certainly stable allocation if and only if the corresponding matching is stable in the related matching market.
\item If $Y$  is a  merit-based stable allocation, then there exists a merit-based certainly stable allocation  $Y'$ that all students weakly prefer which assigns all students to the same colleges (potentially under different contractual terms). 
\end{enumerate}
\end{theorem}

\begin{corollary}[Existence]\label{prop:similar}
The set of merit-based stable allocations is nonempty in Hungarian college admissions markets. In particular, it includes $Y^{\text{\SPDA}}$ and $Y^{\text{\SRDA}}$.
\end{corollary}

\begin{corollary}[Weak rural hospital theorem] \label{prop:rural}
In Hungarian college admissions markets: 

\begin{myenumerate}
\item The set of students assigned to some college is identical under all merit-based stable allocations. Formally, $\left[Y^{\text{\SPDA}}\right]_S=\left[Y\right]_S$ for any merit-based stable allocation $Y$. 
\end{myenumerate}
Moreover, if $Y$ is a certainly stable merit-based allocation, then for each $c\in C$ and $t\in \left\{0,1\right\}$:
\begin{myenumerate}
\item The number of state-funded (self-funded) students assigned to $c$ is equal under $Y^{\text{\SPDA}}$ and $Y$. Formally, $\left|Y^{\text{\SPDA}}\cap \left(S\times \{c\}\times \{t\}\right)\right|=\left|Y\cap \left(S\times \{c\}\times \{t\}\right)\right|$.
\item If college $c$ does not fill one of its quotas under \SPDA, the same students are assigned to $c$ with these financial terms under Y. Formally, if ${\left|Y^{\text{\SPDA}}\cap \left(S\times \{c\}\times \{t\}\right)\right|<q_c^t}$ then $Y^{\text{\SPDA}}\cap \left(S\times \{c\}\times \{t\}\right)=Y\cap \left(S\times \{c\}\times \{t\}\right)$.
\end{myenumerate}
\end{corollary}

Taken together, \cref{prop:similar,prop:rural} show that changing the orientation of DA cannot lead to a change in the number or even the composition of the assigned students.

\begin{corollary}[Weak consensus property] \label{prop:consensus} 
In Hungarian college admissions markets, all students weakly prefer $Y^{\text{\SPDA}}$ to any other merit-based stable allocation. 
\end{corollary}

Taken together, \cref{prop:similar,prop:consensus} imply that it is possible to identify the student-optimal merit-based stable allocation time-efficiently (using the time-efficient \SPDA).




\subsection{Beyond Merit-Based Stable Allocations}\label{sec:beyond_da}
We begin by showing that Hungarian college admissions markets may admit stable allocations where funding is not allocated based on merit, and that the rural hospital theorem and the consensus property do not extend to this domain.

\begin{proposition}\label{prop:cardinality}
In Hungarian college admissions markets:
\begin{enumerate}
	\item There may exist non-merit-based stable allocations. 
	\item Students may disagree on the most preferred stable allocation.
	\item Different stable allocations may have more or fewer assigned students than $Y^{\text{\SPDA}}$.
\end{enumerate}
\end{proposition}

\cref{DifferenrNumber} has already shown a non-merit-based stable allocation that assigns more students to college than the unique merit-based stable allocation. In the merit-based stable allocation, the higher-ranked (rich) student receives funding even though her preferences are not sensitive to financial terms while the lower-ranked (poor) student, who is sensitive to financial terms, remains unassigned. In the non-merit-based stable allocation, the college applies its local market power over the rich student; i.e., the college refuses to allocate state funding to the rich student who is not sensitive to financial terms.

Building on this intuition, we next present a stable algorithm designed to increase the number of admitted students compared to \SRDA. The algorithm relies on the application of local market power over students who are not sensitive to financial terms. We say that a student, $s$, is \textit{not sensitive to financial terms in college $c$} if she considers contracts with $c$ to be nearly perfect substitutes. Formally, for any $c'\in C$ and any $t\in T$, $(s,c,t)\succ_s (s,c',t) \iff (s,c,1-t)\succ_s (s,c',t)$ and $(s,c,1-t)\succ_s \emptyset \iff (s,c,t)\succ_s \emptyset$. A student is \textit{not sensitive to financial terms} if she is not sensitive to financial terms in any college.

\paragraph{Preference Flip Algorithm} Initialize a set $A^{\prime}\subseteq A$, where $A$ is the set of student--college pairs such that the student is not sensitive to financial terms in the college. For each pair in $A^{\prime}$, flip the order of the contracts with the college in the student's original ROL, so that the state-funded contract appears immediately \textit{after} the self-funded contract, and run \SPDA on the resulting problem. If the resulting allocation is certainly stable (with respect to the original preferences), stop and output this allocation. Otherwise, remove some pairs from $A^{\prime}$ and repeat the process.

The description of the preference flip algorithm intentionally leaves two degrees of freedom: the choice of $A^{\prime}$ and the choice of elements to remove from this set. In the empirical application, we implement the preference flip algorithm by initializing $A'=A$ and by removing pairs corresponding to the highest-ranked student in some college such that the pair is in $A'$ and the state-funded contract between them is part of a potentially blocking allocation as described in the definition of certainly stable allocations (or a random pair, if no such pair exists). We note that with this choice the preference flip algorithm is computationally efficient. 

\begin{proposition} \label{prop:algo}
For any initial set $A'\subseteq A$ and any rule regarding the choice of elements to remove from $A'$, the preference flip algorithm results in a stable allocation. Furthermore, the resulting allocation may not allocate funding based on merit. 
\end{proposition}

In the main analysis we use the preference flip algorithm, which only applies local market power over students who are not sensitive to financial terms. These students could plausibly be identified using hard information. Supplemental~\cref{app:narrow_and_broad} explores algorithms that apply local market power more broadly. These algorithms use the observation that a college has market power over students whose best feasible alternative to the funded contract with the college is attending the college without state funding, even if they do not rank both contracts with the college consecutively. Identifying these students in practice would pose a bigger challenge.

In exploring the set of stable allocations, a natural objective is to identify a stable allocation that assigns the largest number of students. We conclude this section by showing that this objective may not be feasible with current computation technology when the number of students and colleges is large.

\begin{theorem}\label{prop:np}
Finding a maximum-size stable (or certainly stable) allocation in Hungarian college admissions markets is NP-hard.
\end{theorem}

\cref{prop:np} determines that finding a maximum-size stable allocation is not computationally feasible. In Supplemental~\cref{app:proofs:online}, we show that it is not even possible to guarantee finding a stable allocation that achieves at least 94\% of the size of the maximum-size stable allocation. Since our time-efficient algorithm is not optimal, the difference (in cardinality) between its result and the allocation that results from \SRDA is a lower bound on the maximal possible difference. Indeed, \cref{app:tab:narrow_mp} presents the results of other stable algorithms which apply market power more broadly and achieve a larger increase in the number of assigned students in our empirical application.

\section{Large Hungarian College Admissions Markets}\label{sec:large_market}
We have shown that Hungarian college admissions markets may admit stable allocations where funding is not allocated based on merit, and that the rural hospital theorem and the consensus property do not extend to this domain. In this section, we ask what is required for these results to hold in large matching markets.

We find that when students' sensitivity to financial terms is heterogeneous (i.e., some students consider both contracts with the same college to be nearly perfect substitutes while others do not), the preference flip algorithm assigns more students to college relative to the outcome of DA (\cref{prop:big}). By contrast, no stable allocation increases the number of students assigned to college when all students consider both contracts with the same college to be nearly perfect substitutes (\cref{prop:insens}) or when all students do not regard different contracts with the same college as close substitutes (\cref{th:ind_alg}). 



\subsection{A Model of Large Hungarian College Admissions Markets}

\paragraph{Random Hungarian college admissions markets} A \textit{random Hungarian college admissions market} is a tuple  $\tilde{{\Gamma}}=\left< S,C, \left\{\succ_c,\gg_c,q^0_c,q^1_c \right\}_{c\in C}, k, \textit{mode} \right>$. It includes the sets of students and colleges, as well as colleges' preferences, rankings, and quotas. A random Hungarian college admissions market induces a Hungarian college admissions market by drawing students' preferences randomly. The randomization is governed by the integer $k$, which determines the number of contracts to be considered, and the parameter \textit{mode}, which determines the mode of randomization.\footnote{We build on the model of \citet{kp2009} that builds on the earlier work of \citet{im2005x,im2015} for large one-to-one matching markets (without contracts). Since we require colleges to have capacities larger than one we follow the notation and conditions of \citet{kp2009}. \citet{hassidim2017need} have shown that similar results hold for other large market models.} 

\paragraph{Modes of randomization} We focus on two main modes of randomization, \textit{independent} and \textit{substitutes}. 
Under the  \textit{independent} mode of randomization, each student's preferences are randomly drawn independently and uniformly from the set of all permutations over $2k$ contracts (other contracts are not acceptable). Under this mode of randomization, a typical student is extremely sensitive to financial terms (i.e., finds at most one set of financial terms acceptable in each college).

Under the \textit{substitutes} mode of randomization, preferences are drawn as follows. First, for each student, $k$ acceptable colleges are drawn uniformly at random. Second, a permutation over the $2k$ contracts with these colleges is drawn independently and uniformly at random. This mode of randomization leads to substantial heterogeneity in students' sensitivity towards financial terms (i.e., some students rank the state-funded contract and the self-funded contract with the same college consecutively, while others do not). 


%


We note that under the \textit{substitutes} mode of randomization, there are no students who are extremely sensitive to financial terms (i.e., rank state-funded contract exclusively). The presence of such students would greatly simplify the proof of \Cref{prop:big} and improve the bounds therein.\footnote{\citet{rheinganslarge} proposes other ways to improve our bounds.} 

All of our results (for both modes of randomization) continue to hold when all (or most) students prefer state-funded seats to self-funded seats in the same program and when students' preferences are drawn from heterogeneous distributions (e.g., some students prefer economics majors while others prefer physics).

\paragraph{Regular sequences of markets} A sequence of random Hungarian college admissions markets, $\left\{\tilde{\Gamma}^n\right\}_{n=1}^{\infty}$ is \textit{regular} if there exist integers $k,\bar{q}$ and $\lambda$, all greater than one, and a mode of randomization $mode$, such that for each $\Gamma^n=\left<  S^n,C^n,\left\{\succ_c,\gg_c,q^0_c,q^1_c \right\}_{c\in C^n}, k^n, \textit{mode}^n \right>$: 

\begin{enumerate}
    \item $\left|C^n\right|=n$ for all $n$,
    \item $k^n=k$ and $\textit{mode}^n=\textit{mode}$ for all $n > 2k$,
    \item $0<q^t_c\leq\bar{q}$ for all $n$, $c\in C^n$ and $t\in T$, 
    \item $s\gg_c\emptyset$ for all $n$, $c\in C^n$, and $s\in S^n$,
    \item $\frac{1}{\lambda} n\leq \left|S^n\right|\leq \lambda n$ for all $n$.
\end{enumerate}

Condition~1 assures that the number of colleges grows as the sequence progresses. Condition~2 assures that the  distribution of preferences is governed by the same mode of randomization and that the number of contracts that students consider acceptable is uniformly bounded on the sequence. Condition~3 assures that the number of positions of each type is positive and uniformly bounded across colleges and markets. Condition~4 assures that colleges find any student acceptable. 
Condition~5 assures that
the number of students does not grow much faster or much slower than the
number of colleges. \citet{kp2009} require only the first half of this condition. We require the second half as well since we are interested
in instances where a substantial fraction of colleges have multiple stable allocations. In markets with a small number of students (who each find contracts with at most $k$ colleges acceptable) most colleges will not even be a party to any individually rational allocation. The only other difference from the conditions of \citet{kp2009} is the requirement that colleges offer both types of seats (Condition~3). Had each college only offered one type of seat, both the {\it independent} and the {\it substitutes} modes of randomization would induce regular (uniform) \citeauthor{kp2009} sequences of markets.\footnote{These sequences would also meet \citeauthor{kp2009}'s thickness condition.} All of our results continue to hold if the fraction of colleges that offer multiple contracts is bounded away from zero.

\subsection{Stable Allocations in Large Markets}

\begin{theorem} \label{prop:big}
For any regular sequence of markets with the \textit{substitutes} mode of randomization, there exists $\Delta>0$, such that comparing the stable allocation $\tilde{Y}'_n$ to the outcome of \SPDA:
\begin{enumerate}
	\item The expected number of assigned students is larger. Formally, $$\underset{n\rightarrow \infty}{\liminf}\frac{\E\left[\left|\tilde{Y}'_n\right|\right]}{\E\left[\left|\tilde{Y}^{\SPDA}_n\right|\right]}>1.$$
	\item The expected fraction of newly assigned students, funding losers, and upward transfers are all bounded below by $\Delta$.
\end{enumerate}
Finally, given access to students' preferences, and colleges' quotas and rankings of students, $\tilde{Y}'_n$ can be computed time-efficiently (i.e., in polynomial-time) using the preference flip algorithm.
\end{theorem}

\medskip

\begin{proposition} \label{prop:insens}
If all students in a Hungarian college admissions market are not sensitive to financial terms, then the same students are assigned across all stable allocations, and each college is assigned the same number of students across all stable allocations. 
\end{proposition}
 
We note that \cref{prop:insens} is not probabilistic and does not rely on any of the large market assumptions.  
 
\begin{proposition} \label{th:ind_alg}
Let $\{\tilde{\Gamma}^n\}_{n=1}^{\infty}$ be a regular sequence of markets such that the mode of randomization is $\textit{independent}$. Then, the expected fraction of colleges  (students) that receive the same assignment in all stable allocations approaches 1 as the  number of colleges, $n$, approaches infinity.
Furthermore, with high probability, the fraction of colleges (students) that receive the same assignment in all stable allocations is arbitrarily close to 1 in sufficiently large markets.
\end{proposition}

\cref{th:ind_alg} states that the set of stable allocations is small when students do not regard different contracts with the same college as close substitutes. This result generalizes to a broad class of modes of randomization that preserve this property (see Supplemental~\Cref{app:proofs:online}). These include modes of randomization where all (or most) students prefer state-funded seats to self-funded seats in the same program as well as modes where students' preferences are drawn from heterogeneous distributions (e.g., some students prefer economics majors while others prefer physics). In our empirical application, a large fraction of applicants rank the state-funded and self-funded contracts with the same college (see \cref{sec:results}). This pattern is not consistent with any of these modes of randomization.

We conclude by noting that our model preserves the property of the \citeauthor{kp2009} model that the outcomes of \SPDA and \SRDA are similar.

\begin{proposition} \label{prop:sub_vs_star1}
Let $\{\tilde{\Gamma}^n\}_{n=1}^{\infty}$ be a regular sequence of markets such that the  mode of randomization is $independent$ or $substitutes$. Then, the expected fraction of colleges (students) that receive a different assignment under $\tilde{Y}^{\SPDA}$ and $\tilde{Y}^{\SRDA}$ approaches zero as the number of colleges, $n$, approaches infinity.
\end{proposition}

\section{Empirical Findings}\label{sec:results}
In this section, we evaluate our theoretical predictions using data from Hungarian college admissions. We document heterogeneity in applicants' sensitivity to financial terms. Consistent with our theory, we find that the preference flip algorithm increases the number of applicants admitted to college relative to the outcome of \SRDA and that \SPDA and \SRDA are essentially identical.\footnote{\Cref{app:ipmm} presents a similar analysis using data from the Israeli Psychology Master's Match, where an overwhelming majority of applicants are not sensitive to financial terms. The preference flip algorithm does not increase the number of admitted applicants.} Additionally, the preference flip algorithm benefits low socioeconomic status applicants and increases geographic mobility, especially to the periphery.

\subsection{Data}\label{sec:data}
Our main data source is an administrative dataset that contains information on the bachelor's degree admissions process in Hungary. In particular, we observe each applicant's ROL, the priority score for each contract in the ROL, and all the information required to (re)calculate these priority scores.\footnote{Our data report up to 7 contracts from each ROL: the first 6 contracts and the contract to which the applicant is assigned. The dataset also reports the number of contracts in each ROL. We observe the complete ROL for 91.7 percent of applicants and for 93.5 percent of all ranked contracts.} This information includes grades in various subjects in the final two years of high school (11th and 12th grades), performance in the matriculation exams, and the number of points the applicant received for claiming a disadvantaged background. For each applicant we also observe gender, postal code, and a high-school identifier. For each alternative (program--terms pair) the data includes its realized priority-score cutoff (i.e., the minimal score that was required to gain admission to this alternative).
We complement our data with hand-collected information on the program-specific capacities from the 2007 information booklet.

The second data source is the T-STAR dataset of the Hungarian Central Statistics Office. We use it to obtain settlement-level annual information on collected income taxes. In particular, we calculate the per-capita gross annual income for all 3,164 settlements in 2007 and merge it with our main administrative data based on settlement identifiers.

The third data source is the National Assessment of Basic Competencies (NABC). Following \citet{horn2013diverging}, we create an NABC-based SES index, which is a standardized measure that utilizes survey information of the NABC. The NABC-based SES index resembles the economic, social, and cultural status (ESCS) indicator of the OECD PISA survey. It combines three subindices: an index of parental education, an index of home possessions (number of bedrooms, mobile phones, cars, computers, books, etc.), and an index of parents' labor-market status. Since this data is available staring in 2008, we calculate the average value of students' NABC-based SES index for each high school between 2008 and 2012, and merge this with our main administrative dataset using high-school identifiers.\footnote{The high-school-specific average NABC-based SES index is stable over time. Between 2008 and 2012, 95\% of the variation in the high-school-specific averages of the NABC-based SES index is explained by high-school fixed effects. Furthermore, each year, about one-third of the student-level variation in the NABC-based SES index is explained by high-school fixed effects.}

\Cref{app:tab:summary_stat_applicants} summarizes the means and standard deviations of the background characteristics of applicants.

To set the stage for our main analysis, we first calculate the outcome of \SRDA. The \SRDA assignment replicates the realized assignment for all but 3.2 percent of the applicants (3,528 applicants).\footnote{There are small gaps and issues that do not allow us to fully replicate the allocation. First, we only observe up to 7 contracts in each ROL (corresponding to the top 6 choices and the realized allocation). Second, we do not know the results of auditions and other specialized admissions exams that are held by a small number of programs (e.g., art programs). Third, the mechanism has rules regarding cases where there are multiple (tied) marginal applicants which are not observable to us \citep{biro2008student,biro2015college}. We address this by breaking priority ties with a single lottery. Finally, the data contain some inconsistencies, for example, applicants with a priority score of zero assigned to selective programs. In our main analysis, we address these inconsistencies by fixing the assignment of these applicants to the realized assignment, and make the corresponding seats unavailable to others. \Cref{app:section:empirical_findings} presents the results from other approaches, and establishes the robustness of our findings.} We call the resulting allocation the benchmark. \cref{app:tab:benchmark} shows that the benchmark allocation is larger by 2,567 assigned applicants, and that the difference is the result of increased utilization of self-funded seats. Intuitively, the higher utilization of self-funded seats reduces the scope for increasing the size of the stable allocation using the preference flip algorithm.

\subsection{The Distribution of Applicants' Rank-Order Lists}

\Cref{tab:summary_stat_rols} presents summary statistics on the characteristics of applicants' ROLs for all applicants, and for disadvantaged and non-disadvantaged applicants separately. Panel A presents statistics related to applicants' sensitivity towards funding. Panel B focuses on program characteristics, such as geographic location and field of study.


Panel A shows that 51.4 percent of applicants were extremely sensitive to financial terms: they ranked state-funded contracts exclusively. Other applicants were less sensitive to financial terms, including 11.6 percent of the applicants who ranked the same program with state-funding and self-funding consecutively on the top of their ROL.

Panel B shows that 49.4 percent of applicants rank exclusively contracts that are in the same settlement, 35 percent of applicants rank exclusively contracts at a single university, and 26.3 percent of applicants rank exclusively contracts at a single faculty of a university. Panel B also shows that 54.9 percent of applicants rank exclusively contracts in a single field of study, and 29.6 percent of applicants rank exclusively contracts in a single major. 

Disadvantaged and non-disadvantaged applicants differ in terms of their sensitivity to financial terms and other program characteristics. Non-disadvantaged applicants were nearly twice as likely to rank the same program with state-funding and self-funding consecutively, and 29.4 percentage points less likely to rank state-funded contracts exclusively compared to disadvantaged applicants. Additionally, disadvantaged applicants were less likely to express preference for program characteristics.



Overall, these patterns suggest that disadvantaged and non-disadvantaged applicants trade off program characteristics (e.g., major or location) and financial terms differently. Disadvantaged applicants are more sensitive to financial terms, while non-disadvantaged applicants are more willing to pay for certain program characteristics. Since market power can only be exercised over applicants who rank two contracts in the same program, these patterns suggest that the direct effect of the preference flip algorithm would likely benefit disadvantaged applicants.

\begin{table}[h!]
\centering\footnotesize
	\captionsetup{justification=centering}
\caption{Characteristics of applicants' ROLs}
\vspace{-1.5em}
\label{tab:summary_stat_rols}
\begin{tabular}{l S[table-number-alignment = center-decimal-marker] S[table-number-alignment = center-decimal-marker] S[table-number-alignment = center-decimal-marker] S[table-number-alignment = center-decimal-marker]} \\ \hline\hline
 & \multicolumn{1}{c}{{\parbox{2.1cm}{\centering All applicants (\%)}}} & \multicolumn{1}{c}{{\parbox{2.1cm}{\centering Non-disadvantaged (\%)}}} & \multicolumn{1}{c}{{\parbox{2.1cm}{\centering Disadvantaged (\%)}}} & \multicolumn{1}{c}{{\parbox{2.1cm}{\centering p-value ((2)=(3))}}} \\ 
 & \multicolumn{1}{c}{(1)} & \multicolumn{1}{c}{(2)} & \multicolumn{1}{c}{(3)} & \multicolumn{1}{c}{(4)} \\ \hline
 \multicolumn{5}{l}{A. {\it Preference for funding}} \\ 
 {\parbox{6.5cm}{State-funded contract exclusively}} & 51.4 & 50.0 & 79.4 &  0.00 \\ 
 {\parbox{6.5cm}{Self-funded contract exclusively}} & 18.5 & 19.2 &  2.0 &  0.00 \\ 
 {\parbox{6.5cm}{State- and self-funded contracts}} & 30.1 & 30.8 &  18.6 &  0.00 \\  
 {\parbox{6.5cm}{Same study program consecutively}} & 15.5 & 15.8 &  8.3 &  0.00 \\ 
 {\parbox{6.5cm}{Same study program consecutively on the top of the ROL}} & 11.6 & 11.9 &  6.2 &  0.00 \\ 
 \multicolumn{5}{c}{} \\ 
 \multicolumn{5}{l}{B. {\it Preference for program characteristics}} \\ 
 {\parbox{6.5cm}{Single program location}} & 49.4 & 50.0 & 36.6 &  0.00 \\ 
 {\parbox{6.5cm}{Single university}} & 35.0 & 35.4 & 27.4 &  0.00 \\ 
 {\parbox{6.5cm}{Single faculty}} & 26.3 & 26.9 & 15.3 &  0.00 \\ 
 {\parbox{6.5cm}{Single field of study}} & 54.9 & 55.6 & 40.1 &  0.00 \\ 
 {\parbox{6.5cm}{Single major}} & 29.6 & 30.3 & 15.5 &  0.00 \\ 
\hline \multicolumn{1}{l}{\# of applicants} & \multicolumn{1}{c}{     108,854} & \multicolumn{1}{c}{     103,840} & \multicolumn{1}{c}{       5,014} & \\ 
\hline\hline \multicolumn{5}{p{17cm}}{{\it Notes}: The table reports summary statistics on applicants' ROL. Panel A focuses on preference for funding. Specifically, Panel A shows the share of applicants who rank state-funded (self-funded) contracts exclusively, the share of students who both rank state-funded and self-funded contracts, and the share of applicants who rank the same study program with state-funding and self-funding consecutively, and who rank the same study program with state-funding and self-funding consecutively on the top of the ROL. Panel B focuses on preference for program characteristics. Specifically, Panel B shows the share of applicants who rank exclusively contracts that are in the same settlement (single program location), at a single university (single university), and at a single faculty of a university (single university), and the share of applicants who rank exclusively contracts in a single field of study, and in a single major. Column (1) presents these shares for all applicants, and columns (2) and (3) report these shares for non-disadvantaged and disadvantaged applicants, respectively. We test whether these shares differ between disadvantaged and non-disadvantaged applicants. Column (4) reports the corresponding p-values.}\end{tabular}
 
\end{table}

\subsection{The Size of the Set of Stable Allocations}
This section compares the stable allocations resulting from \SRDA, \SPDA, and the preference flip algorithm along several dimensions. We show that \SRDA and \SPDA yield essentially identical allocations, while the preference flip algorithm leads to meaningful changes: it changes the assignment of approximately 8 percent of college applicants benefiting low-SES applicants disproportionately, increases the number of applicants admitted to college by about 2 percent, and promotes geographic mobility relative to \SRDA and \SPDA. 

\subsubsection{Winners, Losers, and the Number of Assigned Applicants}

\Cref{tab:core_size} compares the outcome of the preference flip algorithm to the \SRDA benchmark. Consistent with \Cref{prop:big}, Panel A shows that the preference flip algorithm increases the number of assigned applicants by 1.9 percent (1,558 applicants). The increase is concentrated in self-funded, full-time seats. 

In the absence of the consensus property, it is interesting to understand how the preference flip algorithm affects different groups of students. Panel B provides statistics on the set of winners from the preference flip algorithm. We find that the preference flip algorithm generates 4,798 winners relative to \SRDA. Approximately 44 percent of the winners were not assigned under the benchmark (i.e., are newly assigned by the preference flip algorithm) and a similar number are assigned to a different program from that of their benchmark assignment. The rest of the applicants are assigned to the same program under preferable financial terms. 

Panel C provides statistics on the set of losers from the preference flip algorithm. We find that the preference flip algorithm generates 4,260 losers relative to \SRDA. The majority of losers (more than 70 percent) are assigned to the same program under less favorable financial terms. The rest of the losers are assigned to another program under a contract they ranked lower or become unassigned. 		  
Our theoretical analysis shows that the impact of the preference flip algorithm on redistribution is complex (\cref{prop:algo}).
The direct effect of applying local market power is that some applicants lose their state funding and other applicants win by gaining admission to the program with state funding (instead of being unassigned or being assigned to a less desirable alternative). But, if the program already filled its capacity for self-funded applicants, this means that some applicants will lose their assignment to the unfunded contract. Additionally, displaced applicants (either winners or losers) can lead to chains of reassignments. This complexity is reflected in Panels B and C. Our empirical evaluation shows that, in spite of the theoretical possibility, a relatively small number of losers become unassigned or assigned to a new program.

We also evaluate what a single college can achieve by applying local market power. To this end, we focus on programs that can increase the size of their cohort (have unfilled self-funded capacity) and run a version of preference flip algorithm that allows only this program to apply local market power (by initializing $A'=A\cap\left(S\times\{c\}\times T\right)$). We find that, on average, the number of applicants assigned to this program increases by 3.6. This number corresponds to a total increase of 1,741 applicants (compared to 1,558 when $A' = A$).

\begin{table}[h!]
\centering\footnotesize
\captionsetup{justification=centering}
\caption{Winners, losers, and the number of assigned applicants}
\vspace{-1.5em}
\label{tab:core_size}
\begin{tabular}{l S[table-number-alignment = center, table-format = 5.0] S[table-number-alignment = center, table-format = 5.0] S[table-number-alignment = center, table-format = 5.0]} \\ \hline\hline
 & \multicolumn{1}{c}{Benchmark} & & \multicolumn{1}{c}{Preference flip } \\
 & \multicolumn{1}{c}{SR-DA} & \multicolumn{1}{c}{SP-DA} & \multicolumn{1}{c}{algorithm} \\ 
 & \multicolumn{1}{c}{(1)} & \multicolumn{1}{c}{(2)} & \multicolumn{1}{c}{(3)}  \\ \hline
\multicolumn{4}{l}{A. {\it Number of assigned applicants}}\\ 
Assigned to a contract &        84130 &        84130 &                85688 \\ 
Assigned to a state-funded contract &        48725 &        48725 &            48709  \\ 
Assigned to a self-funded contract &        35405 &        35405 &            36979  \\ 
Assigned to a self-funded, full-time contract &        13321 &        13321 &              14418   \\ 
\multicolumn{4}{l}{}\\ 
B. {\it Winners} & &            8&                 4798  \\ 
 -- Newly assigned &  &            0&                2121   \\ 
 -- New program & &            8 &                 2282   \\ 
 -- Same program, preferred financial terms & &            0 &                395   \\ 
\multicolumn{4}{l}{}\\ 
C. {\it Losers} &  &            0 &                4260    \\ 
 -- Newly unassigned & &            0 &             563   \\ 
 -- New program & &            0 &                674   \\ 
 -- Same program, less preferred financial terms &   &            0 &          3023  \\ \hline\hline
\multicolumn{4}{p{15.5cm}}{{\it Notes}: The table presents the number of assigned applicants under various stable assignments (Panel A). Panels B and C describe the number of winners and losers from changing the benchmark \SRDA to \SPDA and to our alternate algorithms. An applicant is a winner if she is assigned to a contract she ranked higher than her contract in the benchmark assignment. An applicant is a loser if she is assigned to a contract she ranked lower than her assigned contract in the benchmark assignment or if she becomes unassigned. Column (1) shows the benchmark assignment. Column (2) presents the \SPDA assignment. Column (3) presents the preference flip algorithm assignment.} \\ 
\end{tabular}

\end{table}

Column (2) of \cref{tab:core_size} presents the results of changing \SRDA to \SPDA. As predicted by \cref{prop:rural}, \SPDA and \SRDA assign the same number of applicants. Furthermore, changing the assignment mechanism from \SRDA to \SPDA only changes the allocation of 8 applicants (approximately 0.01\%). Consistent with \cref{prop:consensus}, all 8 applicants prefer their \SPDA assignment.

According to \cref{prop:dgs}, only the 8 applicants who strictly prefer the outcome of \SPDA to \SRDA have an incentive to misrepresent their preferences under \SRDA. This finding supports our reliance on reported preferences in our interpretation of winners and losers.\footnote{To verify that applicants had limited incentives to misrepresent their preferences, we also conducted a simulation study. We added to each applicants' ROL all contracts with 6 programs in fields that appear on the original ROL. We then drew a random ROL consisting of all these contracts uniformly at random and ran \SPDA and \SRDA. We repeated this process 100 times. On average, 7.1 applicants (out of the 108,854) could manipulate \SRDA. We repeated this analysis requiring that the funded contract in each program is ranked higher than the unfunded contract with the same program in all ROLs. On average, 7.4 applicants could manipulate \SRDA.}

\subsubsection{Characteristics of Winners and Losers}

\Cref{tab:student_comparison_h2} compares the characteristics of winners and losers for the preference flip algorithm. On average, winners' SES is lower than losers'. The share of disadvantaged applicants among winners is 6.0 percent relative to 2.5 percent among losers. Additionally, the average per-capita income in winners' settlements of residence is lower, the average NABC-based SES index in their high schools is higher, and they are less likely to live in the capital or a county capital.

\Cref{tab:student_comparison_h2} also reveals that, on average, winners have lower academic achievement than losers. The average 11th-grade GPA among winners is 3.5, relative to 3.7 among losers. The difference corresponds to 0.26 (=(3.690-3.475)/0.84) standard deviations of the distribution of college applicants' 11th-grade GPA. In summary, applicants who prefer the assignment resulting from the preference flip algorithm to the merit-based benchmark are lower-achieving low-SES applicants. 


\begin{table}[htpb!]
\centering\footnotesize
\captionsetup{justification=centering}
\caption{Characteristics of winners and losers: preference flip algorithm}
\vspace{-1.5em}
\label{tab:student_comparison_h2}
\begin{tabular}{l S[table-number-alignment = center-decimal-marker] S[table-number-alignment = center-decimal-marker] S[table-number-alignment = center-decimal-marker]} \\ \hline\hline
 & \multicolumn{1}{c}{Winners} & \multicolumn{1}{c}{Losers} & \multicolumn{1}{c}{p-values: (1)=(2)} \\ 
 & \multicolumn{1}{c}{(1)} & \multicolumn{1}{c}{(2)} & \multicolumn{1}{c}{(3)}  \\ \hline
 Disadvantaged & 0.060 &  0.025 & 0.000 \\ 
 Per-capita annual gross income (1000 USD, 2007 prices) & 9.719 &  10.399 & 0.000 \\ 
 NABC-based SES index & 0.174 &  0.306 & 0.000 \\ 
 Capital & 0.184 &  0.280 & 0.000 \\ 
 County capital & 0.190 &  0.216 & 0.003 \\ 
 Town & 0.330 &  0.293 & 0.000 \\ 
 Village & 0.297 &  0.212 & 0.000 \\ 
 11th-grade GPA (1--5) & 3.475 &  3.690 & 0.000 \\ 
 Female & 0.571 &  0.552 & 0.076 \\ 
\hline Number of applicants & \multicolumn{1}{c}{       4,798} & \multicolumn{1}{c}{       4,260} & \multicolumn{1}{c}{       9,058} \\ \hline\hline 
\multicolumn{4}{p{15.5cm}}{{\it Notes}: The table reports the mean values of characteristics of winners and losers from changing the benchmark \SRDA to the preference flip algorithm. An applicant is a winner if she is assigned to a contract she ranked higher than her benchmark assignment. An applicant is a loser if she is assigned to a contract she ranked lower than her assigned contract in the benchmark assignment or if she becomes unassigned. Column (3) presents p-values for the equality of mean characteristics of winners and losers.}
\end{tabular}

\end{table}

\subsubsection{Geographic Mobility}



\Cref{tab:mobility2} studies geographic mobility. Geographic mobility is affected through winners who get newly assigned or get assigned to a new program at another location, and through the smaller number of losers who get unassigned or get assigned to a new program at another location. We compare the allocation resulting from the preference flip algorithm to the benchmark across three dimensions of geographic mobility. First, columns (1) and (2) focus on \textit{movers}: applicants whose assigned study program is not located in the county where they reside. Second, columns (3) and (4) focus on applicants who get assigned to a study program in the periphery. Third, columns (5) and (6) focus on the subset of movers who move to the capital, that is, applicants who get assigned to a study program in the capital but do not reside there. Since low-SES applicants are more likely to reside outside the capital, and prestigious study programs are located in the capital, this latter dimension of geographic mobility may be interpreted as a measure of social mobility \citep{Narita}.

Columns (1) and (2) show that the preference flip algorithm increases the number of movers by 1,094 (from 51,221 to 52,315, see Panel A). This increase is largely explained by the increase in the number of newly assigned applicants to a location outside their county of residence (1,330 newly assigned movers, see column (2) of Panel B). The difference is the result of 311 movers who become unassigned, an increase of 48 in the number of movers among winners who are assigned to a new program (=1,494--1,446), and an increase of 27 in the number of movers among losers who are assigned to a new program (=403--376).

Columns (3) and (4) show that the number of applicants assigned to the periphery increases by 1,110 (from 47,845 to 48,955, see Panel A). The increase in the number of applicants assigned to the periphery accounts for 71 percent of the increase in the number of assigned applicants. This finding is in line with the fact that most seats in the capital are filled in the benchmark assignment (see Appendix Table~\ref{tab:excess_capacity}). Among winners, the number of applicants assigned to the periphery increases by 1,269 (=1,460+1,211--1,402), and, among losers, the number of applicants assigned to the periphery decreases by 159 (=339--233--265). Thus, unlike rural hospitals, peripheral colleges can improve the utilization of available seats.

Columns (5) and (6) show that the number of applicants who move to the capital increases by 381 (from 22,444 to 22,825, see Panel A). Thus, applicants from the periphery account for 85 percent (=381/448) of the increase in the number of applicants assigned to the capital. Among winners, the number of applicants who move to the capital increases by 590 (=435+694--539), and, among losers, the number of applicants who move to the capital decreases by 209 (=167--167--209).

\begin{table}[htpb!]
\captionsetup{justification=centering}
\centering\footnotesize
\caption{Geographic mobility: Preference flip algorithm}
\vspace{-1.5em}
\label{tab:mobility2}
\begin{tabular}{l c c c c c c c c} \\ \hline\hline
 & \multicolumn{2}{c}{Mover} & & \multicolumn{2}{c}{Assigned to periphery} & & \multicolumn{2}{c}{Moved to capital} \\ \cline{2-3} \cline{5-6} \cline{8-9}
& \multicolumn{1}{c}{\SRDA} & \multicolumn{1}{c}{PFA} & & \multicolumn{1}{c}{\SRDA} & \multicolumn{1}{c}{PFA} & & \multicolumn{1}{c}{\SRDA} & \multicolumn{1}{c}{PFA} \\ 
 & \multicolumn{1}{c}{(1)} & \multicolumn{1}{c}{(2)} & & \multicolumn{1}{c}{(3)} & \multicolumn{1}{c}{(4)} & & \multicolumn{1}{c}{(5)} & \multicolumn{1}{c}{(6)}\\ \hline
\multicolumn{9}{l}{A. {\it Total}} \\ 
All assigned applicants &       51,221 &       52,315 & &       47,845 &       48,955 & &       22,444 &       22,825 \\ 
\multicolumn{9}{l}{} \\ 
\multicolumn{9}{l}{B. {\it Winners}} \\
Newly assigned (N=       2,121) & -- &        1,330 & & -- &        1,460 & & -- &          435 \\
New program (N=       2,282) &        1,446 &        1,494 & &        1,402 &        1,211 & &          539 &          694 \\
\multicolumn{9}{l}{} \\
\multicolumn{9}{l}{C. {\it Losers}} \\
Newly unassigned (N=         563) &          311 & -- & &          233 & -- & &          167 & -- \\
New program (N=         674) &          376 &          403 & &          265 &          339 & &          209 &          167 \\ \hline\hline 
\multicolumn{9}{p{16cm}}{{\it Notes}: The table presents the number of movers (i.e., applicants whose assigned study program is not located in the county where they reside), of applicants assigned to the periphery (i.e., not to the capital), and of applicants who moved to the capital (i.e., applicants who get assigned to a study program in the capital but do not reside there) under the preference flip algorithm (PFA) and in the benchmark \SRDA assignments. Panel A presents the total number of applicants with these characteristics. Panels B and C focus on winners and losers who are assigned to a different study program.}\end{tabular}

\end{table}

\section{Discussion}\label{sec:discussion}

We study stable allocations in college admissions markets where students can attend the same college under different financial terms. We show that the deferred acceptance algorithm identifies a stable allocation where funding is allocated based on merit and that the set of {\it merit-based} stable allocations is small. By contrast, when students are heterogeneous in the way they trade off program characteristics and contractual terms, the set of stable allocations is large and different stable allocations differ in the number of students assigned to college.

While our focus was on financial aid, our results generalize to setting with other contractual terms. In the context of college admissions, such terms include specialized study tracks (e.g., a business school offering management and accounting tracks) or access to dormitories. Beyond college admissions, financial terms may correspond to different service terms in the military \citep{sonmez2013} or different levels of compensation \citep{niederle2007}.

Our findings raise several important questions. First, \textit{how does one implement alternative stable allocations?} 
This is challenging, because non-merit-based stable mechanisms give academically strong applicants an incentive to overstate their sensitivity to financial terms. One approach for addressing this challenge is to use hard information. For example, governments may use existing administrative data or facilitate the collection of information by colleges through forms like the Free Application for Federal Student Aid (FAFSA). American college applicants who complete this form that determines their eligibility for federal financial aid can choose to forward this information to particular colleges. Similarly, enrollment management and personalized financial aid packages are widespread in the U.S. college admissions market, and private companies are offering financial aid optimization tools using large proprietary datasets \citep{lieber2021price}. 

Another interesting question is \textit{how to compute a stable allocation that supports the objectives of a policymaker.} Policymakers may have different objectives (e.g., maximize the number of students assigned, diversity, etc.). We have shown that---at least for the objective of maximizing the number of assigned  students---this computation may be hard. Still, in college admissions markets, a long waiting period until the release of official result is common, and this may allow solving realistic-size problems in practice. Leveraging our characterization of stability (\cref{lem:stable_new}), integer programming methods in the spirit of \citet{agoston2016integer} and \citet{delorme2019mathematical} are a promising direction.

%
%
%
\singlespacing

\bibliography{references}

\clearpage
\begin{appendices}

\section{Proofs}\label{app:proofs}

\begin{proof}[Proof of \cref{lem:stable_new}]

If $Y$ violates one of Conditions~1--4, then it is not stable. For the other direction, assume $Y$ is not stable. We show that if $Y$ violates none of Conditions~1,2, or 4, then if must violate Condition~3. 

Since $Y$ is individually rational (by Condition~1), the allocation is blocked through some $Z \ne \emptyset$. 
By Condition~4, $Z'$ contains a contract $(\bar{s},c,\bar{t})$ such that $(\bar{s},c,1-\bar{t})\notin Y$.
Furthermore, $Y$ is also blocked through $Z' = \Ch_c(Y \cup Z) \setminus Y$, which contains $(\bar{s},c,\bar{t})$. 

Let $(s,c,t)\in Z'$  be a contract with the highest ranked student according to $\gg_c$ (formally,  $\hat{s}\gg_c s \implies \left(\{\hat{s}\}\times\{c\}\times T\right)\cap Z'=\emptyset$).
By Condition~2, $Y$ is not blocked through $\{(s,c,t)\}$. Therefore,  the allocation $Y$ assigns  $q^t_c$ students to $c$  under the financial terms $t$, and all of them are ranked higher than $s$ according to $\gg_c$. Therefore,  since $(s,c,t)\in \Ch (Y\cup Z')$,  there must exist a contract $(s',c,1-t)\in Z'$  such that $(s',c,t)\in Y$. 
We claim that the allocation must be blocked through  $\{(s,c,t),(s',c,1-t)\}$ (in violation of Condition~3). 

To see this, note that if $\left|Y \cap \left(S \times \left\{c\right\} \times \left\{1-t\right\}\right) \right| < q_c^{1-t}$ (i.e., $c$'s $(1-t)$-quota is not full), then $c$ would prefer to add the contract $(s,c,1-t)$ ($s$ is acceptable to $c$ since all contracts in $Z'$ are), and  swaps in the identities of the recipients of financial aid are less important to the college than the composition of the cohort. Otherwise, both of $c$'s quotas are full. Thus, there exists some $s''$ such that $(s'',c,1-t) \in Y$ and $s'' \notin \left[Ch_c(Y \cup Z')\right]_S$. This holds since $s$ (a new student) must displace someone (both quotas are full), and all students who get assigned  to $c$ under financial terms $t$ in $Y$ are ranked higher than $s$ according to $\gg_c$ (hence, for $(s,c,t)$ to be chosen by $c$, a contract with each of them must also be chosen). We denote by $s^{*}$ be the lowest ranked student according to $\gg_c$ such that $(s^{*},c,1-t) \in Y$ and $s^{*} \notin \left[Ch_c(Y \cup Z')\right]_S$.  Finally, note that $s'\ne s^{*}$, and by the properties of $\succ_c$ we have $Ch_c\left(Y \cup \left\{(s,c,t),(s',c,1-t)\right\}\right)=\left(\left(Y\cap\left(S\times\left\{c\right\}\times T  \right) \right)\cup \left\{(s,c,t),(s',c,1-t)\right\}\right)\setminus \left\{(s',c,t),\left(s^{*},c,1-t\right) \right\}$.  
\end{proof}

\begin{proof}[Proof of \cref{lem:outcome_equiv}]
On the run of \SPDA, colleges never face a choice from a set that includes two contracts with the same student. Consequently, colleges' choices on the run of the algorithm coincide with the union of the choices of the two corresponding auxiliary colleges in the related market. Hence, the corresponding matching of $Y^{\text{\SPDA}}$  in the related market is the student-optimal stable matching.
\end{proof}

\begin{proof}[Proof of \cref{prop:dgs}] 
A student who benefits from misrepresenting her preferences to \SPDA (\SRDA) can do the same in the related market (by \cref{lem:outcome_equiv}). The corollary, therefore, follows from the results of \citet{demange1987}, \citet{df1981}, and \citet{roth1982} for the related market.
\end{proof}

\begin{proof}[Proof of \cref{lem:stable_implies_stable}]
\textbf{If $Y$ corresponds to a stable matching then $Y$ is certainly stable:}
We prove the counterpositive:  that if $Y$ is not certainly stable, then the corresponding matching is not stable.    
Let $Y$ be a feasible allocation. If $Y$ is not certainly stable, then it must violate at least one of Conditions~1--3 of \cref{lem:stable_new} or Condition~4'. If $Y$ violates Condition~1 ($Y$ is not individually rational), then either some student is assigned an unacceptable contract, or some college violates its quotas, or some college is assigned an unacceptable student. In all cases  the corresponding matching also violates individual rationality. If $Y$ violates Condition~2  ($Y$ is blocked through $\left\{\left(s,c,t\right)\right\}$), then either $\left|Y\cap \left(S\times\left\{c\right\}\times\left\{t\right\}\right)\right|<q^t_c$ or there exists $(s',c,t)\in Y$ such that $s\gg_c s'$. In either case, $s$ and $(c,t)$ block the corresponding matching in the related market. If $Y$ violates Condition~3 ($Y$ is blocked through $\left\{\left(s',c,1-t\right),(s,c,t) \right\}$ such that  $(s,c,1-t)\notin Y$ and $(s',c,t) \in Y$), then the student $s'$ and the auxiliary college $(c,1-t)$ block the corresponding matching.
Finally, if an individually rational allocation, $Y$, violates Condition~4' then either there exits a college, $c$ and a pair of students that want to swap their terms in this college (in which case the higher ranked student according to $\gg_c$ and the auxiliary college that this student prefers form a blocking pair in the related market), or there exists a college $c$ where some student prefers to get different terms, and the quota of these terms is not full (in which case, the student and the auxiliary college that she prefers form a blocking pair in the related market).  

\textbf{If $Y$ corresponds to a stable matching in the related market, $Y$ allocates funding based on merit:}
Let $Y$ be a stable allocation that corresponds to a stable matching in the related two-sided matching market. Then, for each $c$ and $t$, if $s$ is matched with $(c,t)$ and $s' \gg_c s$ then $s'$ must weakly prefer her assignment to $(c,t)$.

\textbf{If $Y$ is merit-based certainly stable, then it corresponds to a stable matching in the related market:}
Since $Y$ is stable, the corresponding matching cannot violate individual rationality. Similarly, if it is blocked by $s$ and $(c,t)$ it must be that $(s,c,1-t)\in Y$, since otherwise $Y$ would be blocked by $(s,c,t)$. However, by Condition~4' if  $(s,c,1-t)\in Y$ and $(s,c,t) \succ_s (s,c,1-t)\in Y$ then $q_c^t$ students are assigned to $c$ under $Y$ and since $Y$ is merit-based, all of these students are ranked higher than $s$ according to $\gg_c$. 

\textbf{If $Y$ is merit-based stable, then there exists a merit-based certainly stable allocation that all students weakly prefer:} 
If $Y$ is certainly stable, we are done. Otherwise, Condition~4' is violated. Since funding is allocated based on merit, this means that there exist $s\in S$ and $c\in C$ such that $(s,c,t) \in Y$, $(s,c,1-t)\succ_s (s,c, t)$, and $Y\cap \left( S\times \{c\}\times{1-t} \right)<q_c^{1-t}$ (but the college prefers that $s$ is assigned under $t$). Let $s'$ be the highest ranked such student according to $\gg_c$. Let $Y'$ denote the allocation $Y\setminus \{(s',c,t)\} \bigcup\{(s',c,1-t)\}$. This allocation is stable as otherwise $Y$ is not stable (it is blocked through the same set of contracts, or through the same set of contracts together with $\{(s',c,1-t)\})$. It is merit-based since $s'$ is maximal with respect to $\gg_c$. If $Y’$ is certainly stable, we are done. Otherwise, repeat the process. Since in each iteration all students are weakly better off and one student is (strictly) better off, the process will reach a certainly stable allocation after a finite number of steps (Condition~4' cannot be violated if all students receive their first choices). 

\end{proof}
\begin{proof}[Proof of \cref{prop:similar}]
Follows from \cite{gs1962} by \cref{lem:stable_implies_stable} and \cref{lem:outcome_equiv}.
\end{proof}

\begin{proof}[Proof of \cref{prop:rural}] 
Follows from the rural hospital theorem for the related  market. 
\end{proof}

\begin{proof}[Proof of \cref{prop:consensus}] 
Follows from the consensus property in the related  
market.
\end{proof}

\begin{proof}[Proof of \cref{prop:cardinality}] 
Follows from \cref{DifferenrNumber} in the main text and \cref{DifferenrNumberCont} below. 
\end{proof}

\begin{example} \label{DifferenrNumberCont}
There are  three students, $S=\{r,p,g\}$, and two colleges, $C=\{h,c\}$. College $h$ has two seats, but only one of these seats is state-funded ($q^0_h=1$, $q_h^1=1$), and college $c$ has a single state-funded seat ($q^0_c=0$, $q_c^1=1$). Both colleges rank $r$ first, $p$ second, and $g$ third (i.e., $r\gg_h p\gg_h g$ and  $r\gg_c p\gg_c g$). Students' preferences are
$(r,h,1) \succ_r(r,h,0) \succ_r\emptyset$, $(g,h,1) \succ_g(g,h,0) \succ_g\emptyset$, and $(p,h,1) \succ_p(p,c,1) \succ_p\emptyset$. 


Under $Y^{\text{\SPDA}}=\left\{(r,h,1), (p,c,1), (g,h,0)\right\}$, $r$ and $g$ are assigned to $h$, and $p$ is assigned to $c$. An alternative stable allocation, $\{(r,h,0),(p,h,1)\}$, has the two highest-ranked students, $r$ and $p$, assigned to $h$, while the lowest-ranked student, $g$, remains unassigned. 

\end{example}

\begin{proof}[Proof of \cref{prop:algo}] 
The algorithm terminates since the finite set $A'$  becomes smaller in each iteration.
If it stops while $A'$ is not empty, the output is stable by construction. Otherwise, it is stable as it coincides with \SPDA. 
To see that the resulting stable allocation may not allocate funding based on merit note that it selects the non-merit-based stable allocation in Example~\ref{DifferenrNumber}. 
\end{proof}

\begin{proof}[Proof of \cref{prop:np}] 
We prove that \cref{prop:np} holds even if colleges are restricted to offer no more than one seat under each of the financial terms. Our proof relies on a reduction---we show that an algorithm that identifies a stable (certainly stable) allocation of maximum size quickly (in running time polynomial in the size of the input) can also be used to quickly produce a solution to another problem that is known to be NP-hard. Therefore, identifying such an algorithm proves that $P=NP$. The construction is provided in Supplemental~\cref{app:proofs:online}.
\end{proof}

\subsection*{Proof of \cref{prop:big}}
\begin{definition}\label{Def:E}
For $k+1$  distinct colleges, $c, c'_1, c'_2,\dots ,c'_k\in C^n$, and  $q^1_c+q^1_{c'_1} + q^0_{c'_1}+ 2(k-1)\bar{q}+1$ distinct students, $\{r_i\}_{i=1}^{q^1_c}$, $p$, $\{s^1_i\}_{i=1}^{q^1_{c'_1} + q^0_{c'_1}-1}$,  $w$, and $\{\{s^j_i\}_{i=1}^{2\bar{q}}\}_{j=2}^{k}$ let the event \[E^n\left(c,\{c'_j\}_{j=1}^{k},p,w,\{r_i\}_{i=1}^{q^1_c}, \{s^1_i\}_{i=1}^{q^1_{c'_1} + q^0_{c'_1}-1}, \{\{s^j_i\}_{i=1}^{2\bar{q}}\}_{j=2}^{k}\right)\] denote the case where: 
\begin{enumerate}
    \item $r_1\gg_c r_2\gg_c\dots \gg_c r_{q_c^1}\gg_c p$
    \item $s^1_1\gg_{c'_1} s^1_2\gg_{c'_1} \dots\gg_{c'_1}s^1_{q_{c'_1}^0+q_{c'_1}^1-1}\gg_{c'_1} w$ and $p\gg_{c'_1} w$
    \item $s^j_1\gg_{c'_j} s^j_2\gg_{c'_j} \dots\gg_{c'_j}s^j_{2\bar{q}}\gg_{c'_j} w$ for each $ 2\leq j\leq k$
    \item The only students who find contracts with $c$ acceptable are  $\{r_i\}_{i=1}^{q^1_c}$ and $p$. Formally, for all $s\in S^n\setminus\left(\{r_i\}_{i=1}^{q^1_c} \bigcup \{p\}\right)$ and all $t\in T$,  $\emptyset \succ_s (s,c,t)$. 
    \item The only students who find contracts with $c'_1$ acceptable are $p$, $\{s^1_i\}_{i=1}^{q^1_{c'_1} + q^0_{c'_1}-1}$, and $w$. Formally, $\emptyset \succ_s (s,c'_1,t)$ for all $s\in S^n\setminus\left(\{s^1_i\}_{i=1}^{q^1_{c'_1} + q^0_{c'_1}-1} \bigcup \{p,w\}\right)$ and all $t\in T$. 
        \item For $j\geq 2$, the only students who find contracts with $c'_j$  acceptable are  $\{s^j_i\}_{i=1}^{2\bar{q}}$ and $w$. Formally, $\emptyset \succ_s (s,c'_j,t)$ for all $s\in S^n\setminus\left(\{s^j_i\}_{i=1}^{2\bar{q}} \bigcup \{w\}\right)$ and all $t\in T$. 

    \item Each $r_j$ prefers $(r_j,c,1)$ to $(r_j,c,0)$ and prefers $(r_j,c,0)$ to any other contract. Formally, $(r_j,c,1)\succ_{r_j} (r_j,c,0)\succ_{r_j}(r_j,u,t)$ for each $r_j \in \{r_i\}_{i=1}^{q^1_c}$, and  $u\in C^n\setminus \{c\}$ and $t\in T$.
    \item $p$ prefers  $(p,c,1)$ to $(p,c'_1,1)$ and prefers $(p,c'_1,1)$ to any other contract. Formally, $(p,c,1)\succ_{p} (p,c'_1,1)\succ_{p} z$ for each  $z\in \left(\{p\}\times C^n \times T\right) \setminus \{  (p,c,1),(p,c'_1,1) \}$.
    \item $(w,c'_1,1)\succ_w (w,c'_2,1)\succ_w\dots \succ_w (w,c'_k,1)\succ_w \emptyset$.

\item Each $s^j_i$ prefers $(s^j_i,c'_j,1)$ to $(s^j_i,c'_j,0)$ and prefers $(s^j_i,c'_j,0)$ to any other contract. Formally, $(s^j_i,c'_j,1)\succ_{s^j_i} (s^j_i,c'_j,0)\succ_{s^j_i}(s^j_i,u,t)$ for each $2\leq j\leq k$, $s^j_i \in \{s^j_i\}_{i=1}^{\bar{q}}$, and $u\in C^n\setminus \{c\}$ and $t\in T$.

\end{enumerate}

\end{definition}

\begin{lemma} \label{lem:prob}
There exists $F>0$ so that, so long as Conditions 1--3 of \cref{Def:E} are met, 
 the event $E^n\left(c,\{c'_j\}_{j=1}^{k},p,w,\{r_i\}_{i=1}^{q^1_c}, \{s^1_i\}_{i=1}^{q^1_{c'_1} + q^0_{c'_1}-1}, \{\{s^j_i\}_{i=1}^{2\bar{q}}\}_{j=2}^{k}\right)$ 
is realized with probability  greater than $F /n^{q^1_c+q^1_{c'_1} + q^0_{c'_1}+ 2(k-1)\bar{q}+k+1}$ .    
\end{lemma}
\begin{proof}
First, observe that Conditions 1--3 depend on colleges rankings exclusively, so students' preferences are independent from them. Second, 
it is well-known since \citet{im2005x} that conditions of the form \textit{``$m$ schools are not ranked by some subset of students''} have a probability that is bounded below by a constant when each student is interested in a small number of colleges.\footnote{The probability that $m$ schools are not ranked by any student is greater than $(1-mk/n)^{\lambda n}$ converges to $e^{-\lambda mk}$. } Third, the unconditional probability that a student's preferences meet Condition~7 or 10 is greater than $f/n$ for $f>0$ (that depends on $k$), and the probability conditional of Conditions 4--6 being met is slightly greater.
 Similarly, the unconditional probability for the preferences of $p$ (respectively $w$) to meet Condition~8 (respectively 9) is greater than $f/n^2$ (respectively $f/n^k$) and the probability conditional of Conditions 4--6 being met is slightly greater. Since students' preferences are drawn independently, the product of the lower bounds is a valid lower bound for the probability of the event.    
\end{proof}

\begin{lemma}\label{lem:Markov}
Given a Hungarian college admissions market, with $s\gg_c \emptyset$ for all $s\in S$ and $c\in C$, there are at least $\left|S\right|\cdot(\left|S\right|-1)/4$ pairs of distinct students such that  both are ranked outside the top decile of  $\gg$ by at least 40 percent of the colleges. 
\end{lemma}

\begin{proof}
We prove this combinatorial claim using the probabilistic method. First, note that the fraction of pairs of students who both appear in the bottom 9 deciles of $\gg$ of at least $4n/10$ 
colleges is equal to the probability that this condition is satisfied
by a random pair of students drawn uniformly at random from all possible distinct pairs. Let $\mathbb{1}_c(s,s')$ equal 1 if $s$ and $s'$ are is in the top 10 percent of $\gg_c$ and zero otherwise. Then, if a pair of distinct students $(\mathbf{s},\mathbf{s'})$ is drawn uniformly at random, the expected value of  $\mathbb{1}_c(\mathbf{s},\mathbf{s'})$ is lower than $0.2$ using (using Bonferroni's inequalities). By Markov's inequality 
\[\Pr\left[ \sum_{c\in C} \mathbb{1}_c(\mathbf{s},\mathbf{s'})>0.4\left|C\right| \right]\leq \frac{0.2\left|C\right|}{0.4\left|C\right|}=\frac{1}{2}   \] which completes the proof. 
\end{proof}

\begin{lemma} \label{lem:count}
There exists $M>0$ such that there are more than $M\cdot  n^{q^1_c+q^1_{c'_1} + q^0_{c'_1}+ 2(k-1)\bar{q}+k+2}$ ways to select  $\left(c,\{c'_j\}_{j=1}^{k},p,w,\{r_i\}_{i=1}^{q^1_c}, \{s^1_i\}_{i=1}^{q^1_{c'_1} + q^0_{c'_1}-1}, \left\{\{s^j_i\}_{i=1}^{2\bar{q}}\right\}_{j=2}^{k}\right)$ such that Conditions 1--3 of \cref{Def:E}  are met. 
\end{lemma}

\begin{proof}
By \cref{lem:Markov} there are at least $(n/\lambda-2)^2/4$ ways to choose pairs of students that are ranked outside the top decile by at least $n/4$ colleges. 
Given such a pair of students, let 
 $c$ and $\{c'_j\}_{j=1}^{k}$ be colleges such that both students are ranked at the bottom 90 percent (three are at least $\alpha \cdot n^{k+1}$ ordered ways the choose them, for $\alpha>0$). Label the member of the pair that is ranked lower in $c'_1$ as $w$ and the other student as $p$. Choose  $\{r_i\}_{i=1}^{q^1_c}$, $\{s^1_i\}_{i=1}^{q^1_{c'_1} + q^0_{c'_1}-1}$, and $\{\{s^j_i\}_{i=1}^{2\bar{q}}\}_{j=2}^{k}$ distinct students who appear at the top 10 percent of the corresponding colleges ($c$ for $r_i$'s and $c'_j$ for $s^j_i$'s). There at least $\beta\cdot n^{q^1_c+q^1_{c'_1} + q^0_{c'_1}-1 +2(k-1)\bar{q}}$ possible selections, for some $\beta>0$. Multiplying the three bounds we get the required result. 
\end{proof}

\begin{lemma}\label{lem:initialization}

For any college $c$, and any two different selections, $z$ and $z'$, of \[\left(\{c'_j\}_{j=1}^{k},p,w,\{r_i\}_{i=1}^{q^1_c}, \{s^1_i\}_{i=1}^{q^1_{c'_1} + q^0_{c'_1}-1}, \{\{s^j_i\}_{i=1}^{2\bar{q}}\}_{j=2}^{k}\right),\] the events 
$E^n\left(c,z\right)$ and $E^n\left(c,z'\right)$ are disjoint. Furthermore, there exists a polynomial-time algorithm that identifies the set $\left\{z \mid E^n\left(c,z\right) \text{ is realized}\right\}$.

\end{lemma}

\begin{proof}
We describe a polynomial-time algorithm and show that it must identify the only realized event of the form $E^n\left(c,z\right)$, or determine that no such event is realized.
\paragraph{Algorithm}
\begin{enumerate}
\item Find a student, $s$, that ranks $(s,c,1)$ first and $(s,c',1)$ second, for $c'\ne c$. If no such student exists, declare that the set is empty  and stop. Otherwise, this student is the candidate $p$, and the college $c'$ is the candidate $c'_1$. 
\item Find all students other than the candidate $p$ that rank contracts with $c$. If their number differs from  $q_c^1 +1$ declare that the set is empty  and stop. Otherwise, the lowest ranked one according to $\gg_c$ is the candidate $w$. If one of the other  $q_c^1 $
  students does not rank the funded contract with $c$ first and the unfunded contract with $c$ second,  declare that the set is empty  and stop.  Otherwise, they are the candidate $\left\{r_i\right\}_{i=1}^{q^1_c}$ (with indices determined based on $\gg_c$ and Condition~1 of \cref{Def:E}).
\item If the candidate $w$ ranks some funded contract higher than the funded contract with the candidate $c'_1$, declare that the set is empty  and stop. Otherwise, set the candidate  $\left\{c'_j\right\}_{j=2}^{k}$ as the remaining colleges (with indices determined based on the preferences of the candidate $w$ and Condition~9 of \cref{Def:E}).
\item Verify that the number of students ranking the candidate $c'_j$ is exactly as required in the conditions ($2\bar{q}$ for $j>1$, and $q^1_{c'_1} + q^0_{c'_1}-1$ for the candidate $c'_1$) and that all of them rank the funded contract with $c'_j$ first and the unfunded contract with $c'_j$ second.  If this isn't the case, declare that the set is empty  and stop. Otherwise, set these students as candidate $s_i^j$, with indices determined by Conditions~2 and 3 of \cref{Def:E}. 
\item Declare that the event indexed by the vector of candidates was realized. 
\end{enumerate}
It is clear that if an event of the form $E^n\left(c,z\right)$ is realized the algorithm must identify it. Since the output of the algorithm is a unique event, this proves that  $E^n\left(c,z\right)$ and $E^n\left(c,z'\right)$ are disjoint for different selections of $z$ and $z'$, as required. 
\end{proof}

\begin{lemma}\label{lem:different winner}
For any student $w$, and any two different selections, $z$ and $z'$, of \[\left(c,\{c'_j\}_{j=1}^{k},p,\{r_i\}_{i=1}^{q^1_c}, \{s^1_i\}_{i=1}^{q^1_{c'_1} + q^0_{c'_1}-1}, \{\{s^j_i\}_{i=1}^{2\bar{q}}\}_{j=2}^{k}\right),\] the events 
$E^n\left(w,z\right)$ and $E^n\left(w,z'\right)$ are disjoint.\footnote{We abuse notation for readability. The notation is intended to represent a student playing the role of $w$ in \cref{Def:E}. The same remark applies to \cref{lem:different poor}.}
\end{lemma}
\begin{proof}
Towards contradiction, assume that both events are realized. 
By Condition~9 of \cref{Def:E}, $z$ and $z'$ share the same $c'_1$ (it must be a party to the highest ranked funded contract according to  $w$'s preferences).  
Hence, $z$ and $z'$ share the same $p$ (it is the only student that ranks the funded contract with $c'_1$ second after another funded contract.   
This implies that $z$ and $z'$ share the same $c$ (it is the college that is a party to the first ranked contract of $p$). 
But this is a contradiction to \cref{lem:initialization}, that states that events that share the same selection of $c$ are disjoint. 
\end{proof}

\begin{lemma}\label{lem:different poor}
For any student $p$, and any two different selections, $z$ and $z'$, of \[\left(c,\{c'_j\}_{j=1}^{k},w,\{r_i\}_{i=1}^{q^1_c}, \{s^1_i\}_{i=1}^{q^1_{c'_1} + q^0_{c'_1}-1}, \{\{s^j_i\}_{i=1}^{2\bar{q}}\}_{j=2}^{k}\right),\] the events 
$E^n\left(p,z\right)$ and $E^n\left(p,z'\right)$ are disjoint.
\end{lemma}
\begin{proof}
By Condition~8 of \cref{Def:E}, if both events are realized, then $z$ and $z'$ share the same $c$  (it must be a party to the highest ranked  contract according to  $p$'s preferences). 
But this is a contradiction to \cref{lem:initialization}.  
\end{proof}

\begin{lemma}\label{lem:different r1}
For any student $r_1$, and any two different selections, $z$ and $z'$, of \[\left(c,\{c'_j\}_{j=1}^{k},p,w,\{r_i\}_{i=2}^{q^1_c}, \{s^1_i\}_{i=1}^{q^1_{c'_1} + q^0_{c'_1}-1}, \{\{s^j_i\}_{i=1}^{2\bar{q}}\}_{j=2}^{k}\right),\] the events 
$E^n\left(p,z\right)$ and $E^n\left(p,z'\right)$ are disjoint.
\end{lemma}
\begin{proof}
By Condition~7 of \cref{Def:E}, if both events are realized, then $z$ and $z'$ share the same $c$  (it must be a party to the highest ranked  contract according to  $r_1$'s preferences). 
But this is a contradiction to \cref{lem:initialization}.  
\end{proof}

\begin{proof}[Proof of \cref{prop:big}]
First, note that by \cref{lem:count,lem:prob} the expected number of events $E^n(\cdot)$ that are realized is at least $L\cdot n$ for some $L \in (0,1)$. Second, \cref{lem:initialization} guarantees that fixing $c$, the events $E^n(c,z)$ are disjoint for different selections of $z$ (where $z$ is a shorthand for all other selections needed to an event as in \cref{Def:E}).  This means that the probability that fewer than $Ln/2$ colleges play the role of $c$ in an event $E^n(c,z)$ that is realized is at most $(2-2L)/(2-L)$. In other words, the probability that more than $Ln/2$ colleges  play the role of $c$ in an event $E^n(c,z)$ that is realized is at least $L/(2-L)$. 

Third, note that if the preference-flip algorithm is initialized to include in $A'$ only the contracts of the form  $(r_1,c,1)$ for  events $E^n\left(c,\{c'_j\}_{j=1}^{k},p,w,\{r_i\}_{i=1}^{q^1_c}, \{s^1_i\}_{i=1}^{q^1_{c'_1} + q^0_{c'_1}-1}, \left\{\{s^j_i\}_{i=1}^{2\bar{q}}\right\}_{j=2}^{k}\right)$ that are realized, the resulting allocation differs from $Y^{\SPDA}$ only in that students playing the role of $r_1$ will get $(r_1,c,0)$ instead of $(r_1,c,1)$, students playing the role of $p$ will be assigned $(p,c,1)$ instead of $(p,c'_1,1)$, and students playing the role of $w$ will be assigned $(w_,c'_1,0)$ instead of being unassigned. \cref{lem:different winner} guarantees that $w$'s are different across these events. 
Hence, this algorithm never decreases the number of assigned students, and with probability at least $L/(2-L)$, it increases the number of assigned students by a factor of, at least, $L/2\lambda$. \cref{lem:different poor,lem:different r1} provide a similar guarantee for the upward transfers and funding losers.   

Finally, we note that it is possible to initialize the algorithm in polynomial time (e.g., by applying the algorithm identified in \cref{lem:initialization} $n$ times, once for each college), and that our algorithm will terminate in polynomial time (indeed, after the initialization, it will run \SPDA once, and it is well known that the running-time of \SPDA is polynomial). 
\end{proof}

\begin{remark*}
Our proof relies only on the probability with which the events from \cref{Def:E} are realized. So long as this probability remains sufficiently large (at least for a large fraction of these events) our proof will go through. Examples include modes of randomization that restrict $substitutes$ to rank funded positions higher then unfunded positions as well as modes of randomization with correlation structures as in \citet{kp2009}, where some schools are more popular than others, or students' preferences are drawn from heterogeneous distributions (e.g., some students prefer economics majors while others prefer physics).
\end{remark*}

The proofs of \cref{prop:insens}, \cref{th:ind_alg}, and \cref{prop:sub_vs_star1} are provided in Supplemental~\cref{app:proofs:online}.

\clearpage
\begin{center}
{\bf {\Large Supplemental Material (for Online Publication)}}
\end{center}

\clearpage
\setcounter{table}{0}
\setcounter{figure}{0}
\renewcommand{\thetable}{B\arabic{table}}
\renewcommand{\thefigure}{B\arabic{figure}}

\section{Supplemental Proofs (for Online Publication)}\label{app:proofs:online}

\setcounter{table}{0}
\setcounter{figure}{0}
\renewcommand{\thetable}{C\arabic{table}}
\renewcommand{\thefigure}{C\arabic{figure}}

\begin{proof}[Proof of \cref{prop:np}] 
We will prove that \cref{prop:np} holds even if colleges are restricted to offer no more than one seat under each of the financial terms.  Our proof relies on a reduction---we show that an algorithm that identifies a stable (certainly stable) allocation of maximum size quickly (in running time polynomial in the size of the input) can also be used to quickly produce a solution to another problem that is known to be NP-hard. Therefore, identifying such an algorithm will prove that $P=NP$.

The NP-hard problem that we reduce is a special version of MAX-SMTI, studied in \citet{np2}. The description of the problem is as follows. An instance of \textit{the restricted stable marriage problem with incomplete lists and ties} consists of a set of $n$ men, $M$, 	and a set of $n$ women $W=W^{\sim}\cup W^{\nsim}$, where men in $M$ and women in $W^{\nsim}$ have strict preferences over agents on the other side and remaining unmatched, and women in $W^{\sim}$ are indifferent between two acceptable men and find all other men unacceptable. Given such an instance, $I$, a matching is \textit{weakly stable} if no agent's assignment is unacceptable and there is no pair of agents that strictly prefer one another to their assignment. The problem of deciding whether $I$ admits a stable matching under which all women are matched with men is NP-complete \citep{np2}.

Given and instance $I=\left<M,W^{\sim},W^{\nsim},\left\{\succeq_{i}\right\}_{i\in M\cup W}\right>$ we define \textit{the corresponding Hungarian college admissions market} as follows. The set of students is $S^I=\left\{s_i  \right\}_{i\in M\cup W^{\sim}}$, with elements corresponding to each man and each woman in $W^{\sim}$. The set of colleges $C^I=\left\{c_w\right\}_{w\in W^{\sim}\cup W^{\nsim}}$. For each $w\in W^{\nsim}$, $c_w$ has a single state-funded seat ($q_{c_w}^1=1$ and $q_{c_w}^0=0$),  while other colleges have one state-funded seat and one self-funded seat (i.e., for all $w\in W^{\sim}$, $q_{c_w}^1=1$ and $q_{c_w}^0=1$).   

For each $w\in W^{\nsim}$, $c_w$ ranks students (and the outside option) according to the preferences of the corresponding woman ($s_m\gg_{c_w}s_{m'}$ iff $w$ prefers $m$ to $m'$, $s_m \gg_{c_w} \emptyset$ iff $w$ prefers $m$ to her outside option, and $\emptyset\gg_{c_w} s_{w'}$ for all $w'\in W^{\sim}$). For each $w \in W^{\sim}$, $c_w$ has three acceptable students $s_w\gg_{c_w} s_{p_w} \gg_{c_w} s_{f_w}$ where $p_w$ and $f_w$ are the two men acceptable to $w$ (ranked according to some arbitrary rule). Furthermore, for each $w \in W^{\sim}$,  $s_w$'s most preferred alternative is the state-funded seat in $c_w$ followed by a self-funded seat in this college, with all other alternatives being unacceptable ($\left(s_w,c_w,1\right)\succ_{s_w}\left(s_w,c_w,0\right)\succ_{s_w}\emptyset$, and $\emptyset\succ_{s_w}\left(s_w,c,t\right)$ for all $c \ne c_w$ and all $t\in\{0,1\}$). 

Finally, for each $m\in M$, $s_m$ ranks contracts according to $m$'s preferences (ranking the state-funded alternative over self-funded one) except for contracts corresponding to ${w\in W^{\sim}}$ where only one of the terms is acceptable (depending on the students ranking on $c_w$'s ranking). Formally, $s_m$'s preferences satisfy the following conditions: 
\begin{enumerate}
\item  for each $w\in W^{\sim}$, $(s_m,c_w,1)\succ _{s_m}\emptyset$  only if $s_m$ is ranked second according to $\gg_{c_w}$, and $(s_m,c_w,0) \succ_{s_m} \emptyset$ only if $s_m$ is ranked third according to $\gg_{c_w}$.   
\item for each $w\in W$ and $t$ such that the contracts $(s_m,c_w,t)$ is not unacceptable by the first condition $(s_m,c_w,t) \succ_{s_m} \emptyset$ iff $w$ is acceptable to $m$.
\item for each $w,w'\in W$ and  $t,t'$ such that  the contracts $(s_m,c_w,t)$ and $(s_m,c_{w'},t')$ are not unacceptable by the first condition,  $(s_m,c_w,t)\succ_{s_m}(s_m,c_{w'},t')$ iff $m$ prefers $w$ to $w'$.

\end{enumerate}


%
\begin{lemma}\label{lem:np}
An instance $I$ of the  restricted stable marriage problem with incomplete lists and ties admits a stable matching under which all women are matched with men if and only if the corresponding Hungarian college admissions market  admits a stable allocation of cardinality $|S^I|$ (i.e., where all students are assigned to colleges).  
\end{lemma}

\begin{proof}
Assume $I$ admits a stable matching $\mu$ under which all women are matched with men. Consider the allocation $Y$  consisting of the contracts $\left\{\left(s_{\mu\left(w\right)},c_w,1\right)\right\}_{w\in W^{\nsim}}$ together with the contracts  $\left\{\left(s_{\mu\left(w\right)},c_w,t_{\mu\left(w\right)}\right)\right\}_{w\in W^{\sim}}$ and $\left\{\left(s_w,c_w,1-t_{\mu\left(w\right)}\right)\right\}_{w\in W^{\sim}}$, where $t_{\mu \left(w\right)}$ guarantees that $\left(s_{\mu\left(w\right)},c_w,t_{\mu\left(w\right)}\right) \succ_{s_\mu \left(w\right)}\emptyset$ for each $w\in W^{\sim}$.\footnote{Existence of such financial terms is guaranteed since $w$ must be acceptable to $\mu(w)$ and vice versa, by the stability of $\mu$.} 
Then $|Y|=|S^I|$ and $Y$ is a (certainly) stable allocation. 

The only non-trivial case to verify is that of potential blocks involving colleges with two seats (corresponding to women in $W^{\sim}$). This case follows since for each of these colleges the student involved in $\left(s_{\mu\left(w\right)},c_w,t_{\mu\left(w\right)}\right) \succ_{\mu \left(w\right)}\emptyset$ finds matching with the college under the other financial terms unacceptable. Hence, if the college is assigned its first- and second-ranked students, the college does not prefer any other allocation that is individually rational. Furthermore, if the college is assigned its first- and third-ranked students, then the first ranked student receives state funding (her most preferred contract), and the only individually rational allocation that the college prefers would require her to receive a less preferred alternative (the self-funded seat). 

In the other direction, let $Y$ be a (certainly) stable allocation such that $|Y|=|S^I|$. For each $w\in W^{\sim}$ the student $s_w$ must be assigned to $c_w$ (under some financial terms). This holds since $w$ ranks $s_w$ first, and $s_w$ only ranks the contracts with $c_w$ as acceptable. Furthermore, since $S^I$ is equal to the number of available seats, all seats are assigned.

The matching $\nu$ that assigns $w$ to $m$ if and only if $\left(\left\{s_m\right\}\times\left\{c_w\right\}\times T\right) \cap Y \ne \emptyset $  is individually rational, and it clearly matches every women to a men.\footnote{That $\nu$ is indeed a matching follows by the feasibility of $Y$ and the fact that in any stable allocation colleges with more than one seat must be a side to a contracts with a student $s_w$ for some $w\in W^{\sim}$.} 
We claim that it is also weakly stable. To see this, we note that no $w\in W^{\sim}$ can be involved in a blocking pair, as these women get their (tied) first choice men. But a blocking involving $w\in W^{\nsim}$  and $m\in M$ implies that $Y$ is blocked through $\left\{\left(s_m,c_w,1\right)\right\}$.      
\end{proof}

\begin{proposition}
Finding a maximum-size stable (or certainly stable) allocation in Hungarian college admissions markets is NP-hard, even when colleges offer no more than one seat under each of the financial terms (i.e., $q_c^t\leq1$ for each $c$ and $t$).
\end{proposition}

\begin{proof}
Transforming the restricted stable marriage problem with incomplete lists and ties to the corresponding Hungarian college admissions market clearly requires only a polynomial running time. 
An algorithms that finds a maximum-size stable (or certainly stable) allocation in Hungarian college admissions markets should find a stable allocation that matches all students if and only if one exists. If the algorithm is guaranteed to output a maximum-size stable allocation in polynomial time, transforming the original marriage problem and running the algorithm will give an answer in polynomial time to whether the original marriage problem admits a stable matching under which all women are matched with men (by \cref{lem:np}). And this problem is NP-complete \citep{np2}. 
\end{proof}
This completes the proof of \cref{prop:np}.
\end{proof}

\begin{remark*}
MAX-SMTI is not approximable within a factor of 21/19, unless $P=NP$ \linebreak \citep{np1}. Using this fact and the construction used to prove \cref{lem:np}, one can establish that the maximum-size stable (or certainly stable) allocation in Hungarian college admissions markets is not approximable within a constant factor 15/14, even when colleges offer no more than one seat under each of the financial terms (i.e., $q_c^t \leq 1$ for every $c$ and $t$). 
\end{remark*}

\subsection*{Proof of \cref{prop:insens}}
\begin{lemma}\label{lem:insens}
 Let Y be a stable allocation in a Hungarian college admissions market where all students are not sensitive to funding. Consider a modified market that is identical except that for each college, $c$, quotas are given by $\hat{q}^0_c=q^0_c +q^1_c$ and $\hat{q}^1_c=0$. Then, the allocation $\left\{\left(s,c,0\right)\mid  \exists t\in T,\text{ s.t. }  (s,c,t)\in Y   \right\}$ is stable in the modified market. 
  \end{lemma}
  \begin{proof}
  The modified market is (effectively) a matching without contracts market with responsive preferences. Hence, it is sufficient to consider blocking by pairs and violations of individual rationality (\cref{lem:stable_new}). 
  Individual rationality is guaranteed by the stability of $Y$ and students not being sensitive to funding. Similarly, if the  matching (allocation) is blocked by the pair $s$ and $c$ (through $\{(s,c,0)\}$), then $Y$ is blocked through $\{(s,c,0)\}$ or $\{(s,c,1)\}$. 
\end{proof}
  
\begin{proof} [Proof of \cref{prop:insens}]
 By \cref{lem:insens}, each stable allocation $Y$ has the same number of students assigned to each college as the number of students that are assigned to the corresponding college in a stable matching of a certain matching market without contracts where colleges preferences are responsive. The statement therefore follows from the rural hospital theorem.  
\end{proof}
 
\subsection*{Proof of \cref{th:ind_alg}}

\begin{lemma} \label{lem:only one}
If students rank at most one contract with each college, an allocation is stable if and only if it corresponds to a stable matching in the related market.  
\end{lemma}
\begin{proof}
The ``if'' direction follows from \cref{prop:merit}. 
For the ``only if'' direction, note that if the matching is not individually rational, so is the corresponding allocation. Furthermore, if an individually rational matching is blocked, it must be blocked by a pair $(c,t)$ and $s$, such that $s$ is not matched to  $(c,1-t)$ (by individual rationality) nor to $(c,t)$ (by the definition of blocking). Therefore, the corresponding allocation is blocked through the contract $(s,c,t)$. 
\end{proof}

\begin{lemma}\label{bounds}
Let $Y$ be a stable allocation in  $\left< S,C,\left\{\succ_c,\gg _c, q_c^0, q_c^1\right\}_{c \in C},\{\succ_s\}_{s \in S} \right>$, and let $B$ denote the set of students that consider two contracts with the same college acceptable. Then each $s'\in S\setminus B$ weakly prefers $Y^{\SPDA}\left[ \left< S\setminus B,C,\left\{\succ_c,\gg _c, q_c^0, q_c^1\right\}_{c \in C},\{\succ_s\}_{s \in S\setminus B} \right>\right]$ to $Y$, and prefers $Y$ to the allocation $Y^{\SRDA}\left[ \left< S\setminus B,C,\left\{\succ_c,\gg _c, \max\left\{0,q_c^t-\left|\left\{s \in B  \mid (s,c,t)\succ_s \emptyset \right\}  \right|\right\}\right\}_{c \in C, t\in \{0,1\}},\{\succ_s\}_{s \in S\setminus B} \right>\right]$. Furthermore, the opposite comparisons hold for colleges that are not ranked by students in $B$. 
\end{lemma}
\begin{proof}
Consider 
$\left< S\setminus B,C,\left\{\succ_c,\gg _c, \max\left\{0,q_c^t-\left|\left\{s \in B  \mid (s,c,t)\succ_s \emptyset \right\}  \right|\right\}\right\}_{c \in C, t\in \{0,1\}},\{\succ_s\}_{s \in S\setminus B} \right>$, the market where students in $B$ are removed and the quotas are reduced at contracts they rank. $Y  \restriction _{S \setminus B}$ must be stable in this market (otherwise, $Y$ is not stable). Furthermore, since students in $S \setminus B$ do not rank more than one contract with the same college, the allocation must also be stable in the related matching market (by \cref{lem:only one}). Hence, by the consensus property, all students in $S\setminus B$ prefer the outcome of \SPDA in this market to $Y\restriction _{S \setminus B} $ to the outcome of \SRDA (and the opposite holds for colleges not ranked by students in $B$). Furthermore, by Theorem 2.25 in \citet{rs1990}, all students in $S\setminus B$ prefer the outcome of \SRDA in $\left< S\setminus B,C,\left\{\succ_c,\gg _c, q_c^0, q_c^1\right\}_{c \in C},\{\succ_s\}_{s \in S\setminus B} \right>$ to  the outcome of \SRDA in $\left< S\setminus B,C,\left\{\succ_c,\gg _c, \max\left\{0,q_c^t-\left|\left\{s \in B  \mid (s,c,t)\succ_s \emptyset \right\}  \right|\right\}\right\}_{c \in C, t\in \{0,1\}},\{\succ_s\}_{s \in S\setminus B} \right>$ (and the opposite holds for colleges not ranked by students in $B$).
\end{proof}

\begin{lemma}[\citeauthor{kp2009}]\label{lem:dropping}
Consider a matching (without contracts) market and a college $c$ that is matched with $\mu(c)$ under \SPDA. If for any non-empty subset of $\mu(c)$, dropping this subset does not result in an additional offer to $c$ relative to \SPDA, then $c$ is matched with $\mu(c)$ also under \SRDA.
\end{lemma}

\paragraph{Randomized $(\hat{c},\hat{t})$ Algorithm}\label{subs:randomized_algo}

Building on \cref{bounds,lem:dropping}, we present a randomized algorithm that declares failure with probability higher than the probability that there are stable allocations $Y$ and $Y'$ such that $Y\cap S\times \{\hat{c}\}\times\{\hat{t}\}\ne Y'\cap S\times \{\hat{c}\}\times\{\hat{t}\}$ (i.e., different sets of agents are assigned to $\hat{c}$ under $\hat{t}$ in different stable allocations). 
Roughly, the algorithm randomly draws the identities and preferences of students in $B$ (the set defined in \cref{bounds}) and fails if $|B|$ is unusually large or if one of its members ranks a contract with $\hat{c}$. If the algorithm doesn't fail, agents in $B$ are removed and the algorithm proceeds to mimic the distribution of outcomes of \SPDA, where students preferences are drawn as needed based on the distribution that conditions on them not being members of $B$.  Next, the algorithm transitions from the outcome of \SPDA in this market, to a market where quotas are reduced as in \cref{bounds}. The algorithm declares failure if $\hat{c}$ receives an offer in the process (a necessary condition for the assignment of $(\hat{c},\hat{t})$ to change) or if the process induces an ``unusually long" chain of offers. 

\begin{enumerate}
	\item[Step 1.] {\bf Drawing the set $B$ and their preferences:} For each student in $S$ draw preferences according to the $independent$ mode of randomization. 
Denote by $B$ the (random) set of students who rank two contracts with the same college. If a student in $B$ ranks a contract with $\hat{c}$, or if $|B|>\left(1+\log^2(|C|)\right) \cdot \left|S\right|\cdot  8k^2/|C|$ declare failure and stop.  Otherwise, ``forget'' the preferences of students in $S \setminus B$ and continue to Step $2$.
	\item[Step 2.] {\bf \SPDA excluding $B$:}
\begin{enumerate}
	\item[Step 2.1] For each $s\in S \setminus B$,  draw independently and uniformly at random a college from $C$, and draw independently and uniformly at random terms from $T=\{0,1\}$ with equal probability for $0$ and $1$. Have each of these students propose the auxiliary college (college-terms pair) she drew. Auxiliary colleges accept temporarily the most preferred subset of students and reject all others. 
	\item[Step 2.j j>1] If all students in $ S \setminus B$ are temporarily assigned, or have been rejected $2k$ times,  set $Y^1$ to be the resulting allocation and proceed to Step $3$. Otherwise, for each student that is not temporarily assigned and has not been rejected $2k$ times, draw independently and uniformly at random a college in $C$ to which the students did not apply in a previous step of the algorithm (under any terms), and random terms from $T=\{0,1\}$ with equal probability for $0$ and $1$. Have each student propose the auxiliary college she drew. Auxiliary colleges accept temporarily the most preferred subset of the union of the set of new proposal and the proposals they temporarily hold and reject all others. Proceed to Step $2.j+1$. 
\end{enumerate}
\item[Step 3.] {\bf from \SPDA excluding $B$ to \SPDA excluding $B$ with reduced quotas)} Set $Y=Y^1$. Let $\left(s_m,c_m,t_m\right)_{m=1}^{2k|B|}$ be the set of contracts ranked by students in $B$. For each $m$ from 1 to $2k|B|$:
\begin{enumerate}
	\item[Step 3.m] Reduce $q^{t_m}_{c_m}$ by 1 (unless it is already 0). If this does not result in $c_m$ violating its quota under $Y$,  proceed to step $3.m+1$.  Otherwise, reject the lowest ranked student, $s$, assigned to the auxiliary college $(c_m,t_m)$, and proceed to Step $3.m.1$.
	\item[Step 3.m.j] Denote by $P_1>0 $ a parameter (we will specify a value for $P_1$ at a later point in the proof (see \cref{fn:P_1}).
If $j>P_1\cdot \log n$ declare failure and stop.  
If $s$ has been rejected $2k$ times, proceed to step $3.m+1$. Otherwise draw independently and uniformly at random a college to which the student did not apply in a previous step of the algorithm, $c$, and random terms from $t\in\{0,1\}$ with equal probability for $0$ and $1$. If $c=\hat{c}$ declare failure and stop.
Have $s$  propose to the corresponding auxiliary college, $(c,t)$. The auxiliary college accepts temporarily the most preferred subset of $\left(Y\cap\left(S\times \{c\} \times \{t\}  \right)\right) \cup \left\{(s,c,t)\right\}$, and $Y$ is updated accordingly to $\Ch_{(c,t)}\left(Y \cup \left\{(s,c,t)\right\}\right)\cup \left(Y\setminus \left(S\times \{c\}\times \{t\}  \right)\right)$. If the auxiliary college rejects some student (at most one students can be rejected), set $s$ to be this student, and proceed to Step~$3.m.j+1$. Otherwise, proceed to Step~$3.m+1$. 
\item[Step 3.(2k|B|+1)] Set $Y^2$ to be the resulting allocation. proceed to Step 4.
\end{enumerate}
\item[Step 4.] {\bf Stochastic rejection chains:} We note that $\left|Y\cap \left(S\times\{\hat{c}\}\times\{\hat{t}\}  \right)\right|\leq \overline{q}$, and thus this set has no more than $2^{\overline{q}}$ different subsets. 
For each nonempty subset  $Z=\left\{(s_i,\hat{c},\hat{t})\right\}_{i=1}^{l_{Z}}$ of $Y\cap \left(S\times\{\hat{c}\}\times\{\hat{t}\}  \right)$:
\begin{enumerate}
	\item[Step 4.Z] Set $Y=Y^2$ and realized preferences to the part that was drawn  by the end of Step~$3$.
Have $(\hat{c},\hat{t})$ drop all students in $Z_S$ from its preference list (at most ${\overline{q}}$ students) and reject them from the college.  Set $i=b=1$, and proceed as follows:
	\item[Step 4.Z.i.b] If $i=l_{Z}+1$, proceed to the next subset. 
If $b>P_1 \log n$ declare failure and stop.
If $s_i$ has been rejected $2k$ times, move  to step $4.Z.i+1.1$. Otherwise draw independently and uniformly at random a college to which the student did not apply in a previous step of the algorithm, $c$, and random terms  $t\in \{0,1\}$ with equal probability for $0$ and $1$. If the college drawn is $\hat{c}$, declare failure and stop. Otherwise let the chosen auxiliary college, $(c,t)$, choose from $Y \cup \left\{(s_i,c,t)\right\}$. 

If $(c,t)$ does not reject any student, move to Step~$4.Z.i+1$. Otherwise, set $s_i$ to be the rejected student. Update $Y$ according to the choice of $(c,t)$. Proceed to Step~$4.Z.i.b+1$.

\end{enumerate}
	\item[Termination] After all possible subsets of $Z$ (fewer than $2^{\overline{q}}$) have been exhausted, stop. 
\end{enumerate}

\begin{lemma}\label{lem:what to bound}
For any $(\hat{c},\hat{t})\in C^n\times T$, the probability that $(\hat{c},\hat{t})$ has a unique assignment in all stable allocations is bounded below by the probability that the $(\hat{c},\hat{t})$ randomized algorithm terminates without declaring failure.  
\end{lemma}

\begin{proof}
By construction: 
\begin{itemize}
    \item The probability that the algorithm fails in Step $1$ is equal to the probability that the set $B$ has certain properties (regardless of whether or not $(\hat{c},\hat{t})$ has a unique assignment in all stable allocations). 
    \item The probability that the algorithm fails in Step $3$ conditional on not failing in Step $1$ is equal to the probability -- conditional on $B$ not having the aforementioned properties --  that reducing the quotas changes the \SPDA  assignment of $(\hat{c},\hat{t})$ or leads to long rejection chains. 
    \item The probability that the algorithm fails in Step $4$ conditional on not failing in earlier step is greater than the probability -- conditional on $B$ not having the aforementioned properties and quota reductions not changing the \SPDA  assignment of $(\hat{c},\hat{t})$ or leading to long rejection chains -- that the \SPDA and \SRDA assignments of $(\hat{c},\hat{t})$ under reduced quotas are equal (using Boole's inequality and \cref{lem:dropping}). 
\end{itemize}

The result therefore follows by \cref{bounds}.  
\end{proof}

\begin{lemma}[Typical Size of $B$]\label{lem:size of B}
Let $\{\tilde{\Gamma}^n\}_{n=1}^{\infty}$ be a regular sequence of markets such that the mode of randomization is $\textit{independent}$, and let $\tilde{B}^n$ denote the set of students who rank two contract with the same college. Then the probability that  $\left|\tilde{B}^n\right|> \left(1+\log^2(n)\right) \times \left|S^n\right|\times  8k^2/n$    is lower than $1/n^2$ for sufficiently large $n$.   
\end{lemma}

\begin{proof}
When $n>4k$ the probability that an agent finds more than one contract with some college acceptable is 
\[1-1\times \left(1-\frac{1}{2n-1}\right) \times \left(1-\frac{2}{2n-2}\right) \times \dots \times \left(1-\frac{2k-1}{2n-(2k-1)}\right)\leq 1-\left(1-\frac{4k}{n}\right)^{2k}\]
which, by Bernoulli's inequality,  is bounded above by
\[ 1-\left(1-\frac{8k^2}{n}\right)=\frac{8k^2}{n}.\] 
Since students' preferences are drawn independently of one another, the events in which one of them finds more than one contract with the same college acceptable are independent. Hence, using a multiplicative Chernoff bound, this implies that the probability that more than $\left(1+\log^2(n)\right) \times \left|S^n\right|\times  8k^2/n$ students find more than one contract with the same college acceptable is lower than \[\exp\left( -\log^2(n) \cdot   4k^2 \cdot \left|S^n\right|/n \right) \leq\exp\left( -\log^2(n) \cdot   4\lambda k^2   \right), \] 
and 
\[\underset{n\rightarrow \infty}{\lim} n^2\cdot \exp\left( -\log^2(n) \cdot   4\lambda k^2   \right) =0. \] 
\end{proof}

\begin{lemma}\label{lem:c_hat}
Given a regular sequence of markets, there exists $n'$ and $M>0$ such that for all $n>n'$ and any college $\hat{c}\in C^n$, the probability that a contract with $\hat{c}$ is ranked by any student in $B^n$ is bounded above by $M\log ^2(n) / n$. 
\end{lemma}

\begin{proof}
Unless $B$ is larger than $\left(1+\log^2(n)\right) \cdot \left|S^n\right|\cdot  8k^2/n$ (which, for large $n$, occurs with probability lower than $1/n^2$ by \cref{lem:size of B}) the number of contracts ranked by agents in $B$ is lower than $2k\cdot \left(1+\log^2(n)\right) \cdot \left|S^n\right| \cdot 8k^2/n$. Each of these contracts has probability of $1/n$ to be with $\hat{c}$. The result follows by Boole's inequality and the fact that $log ^2(n) / n>1/n^2$ for large $n$. 
\end{proof}

\paragraph{The Effect of Quotas Reduction on \SPDA }

\begin{lemma}[\citeauthor{im2005x}] \label{lem: IM} 
Given a regular sequence of markets, there exists $\alpha >0$ such that for any $n$, the probability that more than $\alpha n$ colleges in $C^n$ are not acceptable to any student in $S \setminus B$ is greater than $1-1/n^2$ under both the $independent$ and the $substitutes$ modes of randomization. 
\end{lemma}

\begin{proof}
The probability that a particular college is ranked by some student is  smaller than $1- (1-k/n)^{n/\lambda}$. Therefore, the expected number of colleges ranked by some student is lower than $M(n):=n\left( 1- \left(1-k/n\right)^{n/\lambda}\right)\geq n\left( 1-\exp(-k/\lambda) \right)$. The indicators of these events meet the conditions of Theorem 1.1 of \citet{impagliazzo2010constructive} (they are negatively correlated), and therefore the probability that  more than $(1+\epsilon)M(n)$ of them are realized is bounded above by $\exp(-2n\epsilon^2)$. 
\end{proof}

\begin{proof}[Proof of \cref{th:ind_alg}]
By \cref{lem:what to bound}, the probability that 
the Randomized $(\hat{c}, \hat{t})$ Algorithm fails is greater than the probability with which $Y\cap\left(S\times\{\hat{c}\}\times\{\hat{t}\}\right)$ varies for different choices of a stable allocation $Y$. 
 By \cref{lem:size of B,lem:c_hat} the probability that it fails in Step $1$ is lower than $A \cdot \log^2 (n) /n $ for sufficiently large $A>0$. 
Conditional on not failing in Step $1$ the arguments of \citeauthor{im2005x} (and \citeauthor{kp2009}) show that the probability that the algorithm fails is bounded above by $A' \cdot \log (n) /n$ for sufficiently large $A'>0$.\footnote{\label{fn:P_1}This argument is now standard in the literature \citep[e.g.][]{im2005x,kp2009, kpr2013, abh2014} and so we omit it for brevity. To summarize: 
by \cref{lem: IM} there exists  excess capacity in more than $\alpha n$ colleges with probability greater than $1-1/n^2$. The probability of failure due to long rejection chain is bounded above by the sum of the probability that there is no such excess capacity and the probability that a geometric random variable reaches values larger than $P_1\log (n)$ for a constant $P_1>0$ (conditional on not stopping, the probability of a chain to stop in the next step continues to be at least $\alpha$). The probability of a geometric random variable to exceed $\log(n)/\alpha$ steps is 
\[(1-\alpha)^{\log(n)/\alpha}=\left(\left(\left(1-\alpha \right)\right)^{1/\alpha}\right)^{\log(n)}<\frac{1}{e}^{\log(n)}=\frac{1}{n}\]
 In the algorithm, we choose the parameter $P_1\geq 1/\alpha$, allowing us to invoke this bound. 
Conditional on no failure  due to a long rejection chain, the probability of failure (since  an offer to $c$ was made) is lower than $2 P_1\log (n)/n$.} Furthermore, the probability that a student has multiple stable assignments is lower than the probability that she ranks a contract with a college that has multiple stable assignments, which is lower than $2k$ times the probability that a college has multiple stable assignments.

Finally, let $X_{(c,t)}$ denote the random variable that the set of students assigned to college $c$ under $t$ differs across stable allocations. 
By the discussion above, the exists $\Delta>0$ such that $\E\left[\sum_{(c,t)\in C\times T} X_{(c,t)}\right] < \Delta \log ^2 n $. Hence, by Markov's inequality 
\[ \Pr\left\{   \sum_{(c,t)\in C\times T} X_{(c,t)} > \Delta \sqrt{n} \log ^2 n   \right\}<\frac{1}{\sqrt{n}}.
\]
The proof for students is completely analogous. 
\end{proof}

\begin{remark*}
Our proof continues to hold under conditions that guarantee that \citeauthor{im2005x}'s short rejection chains argument remains valid. Examples include the variant of the $independent$ mode of randomization where students systematically prefer state-funded seats to self-funded seats in the same program as well as modes where students' preferences are drawn from heterogeneous distributions (e.g., some students prefer economics majors while others prefer physics).
\end{remark*}

\subsection*{Proof of \cref{prop:sub_vs_star1}}
\begin{proof} [Sketch of Proof of \cref{prop:sub_vs_star1}]
First, relying on \cref{lem:outcome_equiv}, we consider the related markets. Second, we note that under the $independent$ mode of randomization, the related markets meet the conditions of \citet{kp2009}.

Next, we consider the $substitues$ mode of randomization.
This case is slightly more complicated. The added complexity arises from the fact that it is likely that when an auxiliary college rejects a student, this student will displace another student at the ``sibling'' auxiliary college, who will go on to propose to the auxiliary college--resulting in an additional proposal. We note, however, that \cref{lem:dropping} can be strengthened, to require that the additional offer to the (auxiliary) college comes from a student that is preferred to all rejected students. Furthermore, students who get displaced from the ``sibling'' auxiliary college by strategically rejected students are not ranked sufficiently high to meet this condition themselves (as otherwise, they would not be displaced) and indirect chains of rejection are unlikely as in the argument of  \citet{im2005x} and \cite{kp2009} (see \cref{lem: IM} and \cref{fn:P_1} for details of the argument).
\end{proof}

\begin{remark*}
The proof extends to modes of randomization with correlation structures as in \citet{kp2009}, where some schools are more popular than others, or students' preferences are drawn from heterogeneous distributions (e.g., some students prefer economics majors while others prefer physics). The proof for modes of randomization that restrict $independent$ or $substitutes$ to rank funded positions higher then unfunded positions is simpler, as it rules out the complication of short rejection cycles discussed above.
\end{remark*}

\clearpage
\section{Supplemental Tables (for Online Publication)}
\Cref{app:sec:data} compares the realized and the benchmark assignments. \Cref{app:sec:summary_stat} presents additional summary statistics.

\subsection{Data}\label{app:sec:data}
This Appendix presents the number of applicants assigned in 2007 and in our benchmark assignment (\Cref{app:tab:benchmark}).

\begin{table}[htpb!!]
\centering\footnotesize
	\captionsetup{justification=centering}
\caption{Realized and benchmark assignments}
\vspace{-1.5em}
\label{app:tab:benchmark}
\begin{tabular}{l c c} \\ \hline\hline
 & Realized & \multicolumn{1}{c}{Benchmark} \\ 
 & \multicolumn{1}{c}{(1)} & \multicolumn{1}{c}{(2)}  \\ \hline
Assigned to a contract &       81,563 &       84,130 \\ 
Assigned to a state-funded contract &       48,726 &       48,725  \\  
Assigned to a self-funded contract &       32,837 &       35,405  \\ \hline\hline
\multicolumn{3}{p{10.5cm}}{{\it Notes}: The table presents the number of applicants assigned under each of the financial terms in the realized assignment in 2007 (column (1)) and in our benchmark assignment (column (2)).} \\ 
\end{tabular}
 
\end{table}

\subsection{Summary Statistics}\label{app:sec:summary_stat}
This appendix presents additional summary statistics. \Cref{app:tab:summary_stat_applicants} summarizes the means and standard deviations of the applicants' background characteristics. \Cref{app:tab:rol_characteristics_residence} provides summary statistics on the characteristics of applicants' ROLs by the type of the settlement where they reside -- an additional proxy for socioeconomic status (cf. \cref{tab:summary_stat_rols}).

Tables~\ref{tab:ses_ROL} and \ref{app:tab:mp_ROL} present the coefficients of a linear regression of ROL characteristics on disadvantaged status and 11th-grade GPA on the sample of applicants. We find that disadvantaged applicants are more (less) likely to rank state-funded (self-funded) contracts exclusively in their ROLs, and they are less likely to rank both contracts with the same study program consecutively. We also find that these differences cannot be explained by the differences in applicants' academic achievement, and thus, by the differences in their admission chances.

Next, we consider three alternative proxies for socioeconomic status: per-capita annual gross income, NABC-based SES index, and the type of the settlement where applicants reside (capital, county capital, town, and village). Tables~\ref{tab:ses_ROL_robust} and \ref{app:tab:mp_ROL_robust} show that conditional on academic achievement, applicants of higher socioeconomic status are more likely to rank at least one self-funded contract in their ROL and they are more likely to rank two contracts with the same study program consecutively. 

\clearpage
\begin{table}[h!]
\centering\footnotesize
\captionsetup{justification=centering}
\caption{Summary statistics on applicants' characteristics}
\vspace{-1.5em}
\label{app:tab:summary_stat_applicants}
\begin{tabular}{l c c c c} \\ \hline\hline
 & \multicolumn{1}{c}{Mean} & \multicolumn{1}{c}{SD} & \multicolumn{1}{c}{N} \\ 
 & \multicolumn{1}{c}{(1)} & \multicolumn{1}{c}{(2)} & \multicolumn{1}{c}{(3)} \\ \hline
 Disadvantaged & 0.05 & 0.21 &       108,854 \\ 
 Per-capita annual gross income (1000 USD, 2007 prices) & 9.90 & 2.30 &       106,934 \\ 
 NABC-based SES index & 0.22 & 0.37 &        84,455 \\ 
 Capital & 0.20 & 0.40 &       106,934 \\ 
 County capital & 0.21 & 0.41 &       106,934 \\ 
 Town & 0.33 & 0.47 &       106,934 \\ 
 Village & 0.26 & 0.44 &       106,934 \\ 
 11th-grade GPA (1--5) & 3.63 & 0.84 &        85,811 \\ 
 Female & 0.57 & 0.50 &       108,854 \\ 
 Number of alternatives in ROL & 3.71 & 2.21 &       108,854 \\ 
 Number of programs in ROL (observed) & 3.01 & 1.30 &       108,854 \\ 
\hline\hline \multicolumn{4}{p{13cm}}{{\it Notes}: The table reports mean values and standard deviations of applicant characteristics. Disadvantaged status is an indicator for claiming priority points for disadvantaged status. 11th-grade GPA is the average grades in mathematics, history, and Hungarian grammar and literature. Applicants' settlement of residence is missing for $       1,920$ applicants, NABC-based SES index is missing for $      24,399$ applicants, and 11th-grade GPA is missing for       23,043 applicants.}\end{tabular}
 
\end{table}

\clearpage
\begin{landscape}
\begin{table}[htp!]
\centering\footnotesize
\caption{Summary statistics on applicants' ROLs: Applicants' residence}
\label{app:tab:rol_characteristics_residence}
\vspace{-1.5em}
\begin{tabular}{l S[table-number-alignment = center-decimal-marker] S[table-number-alignment = center-decimal-marker] S[table-number-alignment = center-decimal-marker] S[table-number-alignment = center-decimal-marker] S[table-number-alignment = center-decimal-marker]} \\ \hline\hline
 & \multicolumn{1}{c}{{\parbox{2.1cm}{\centering All applicants (\%)}}} & \multicolumn{1}{c}{Capital (\%)} & \multicolumn{1}{c}{County town (\%)} & \multicolumn{1}{c}{Town (\%)} & \multicolumn{1}{c}{Village (\%)} \\ 
 & \multicolumn{1}{c}{(1)} & \multicolumn{1}{c}{(2)} & \multicolumn{1}{c}{(3)} & \multicolumn{1}{c}{(4)} & \multicolumn{1}{c}{(5)} \\ \hline
 \multicolumn{6}{l}{A. {\it Preference for funding}} \\ 
 State-funded contract exclusively & 51.4 & 37.1 & 52.5 & 54.3 & 58.2 \\ 
 Self-funded contract exclusively & 18.5 & 24.0 & 18.8 & 17.5 & 15.2 \\ 
State- and self-funded contracts & 30.1 & 38.9 & 28.7 & 28.2 & 26.6 \\ 
 Same study program consecutively & 15.5 & 20.5 & 14.6 & 14.3 & 13.5 \\ 
 Same study program consecutively on the top of the ROL & 11.6 & 15.0 & 11.3 & 10.7 & 10.2 \\  

 \multicolumn{6}{c}{} \\ 
 \multicolumn{6}{l}{B. {\it Preference for study characteristics}} \\ 
  Single program location & 49.4 & 64.5 & 50.0 & 44.4 & 43.5 \\ 
 Single university & 35.0 & 28.8 & 43.2 & 33.7 & 34.7 \\ 
 Single faculty & 26.3 & 24.4 & 30.4 & 25.2 & 25.8 \\ 
 Single field of study & 54.9 & 58.9 & 55.2 & 54.1 & 52.4 \\ 
 Single major & 29.6 & 33.5 & 30.5 & 28.7 & 26.9 \\ 
\hline \# of applicants & \multicolumn{1}{c}{     108,854} & \multicolumn{1}{c}{      21,731} & \multicolumn{1}{c}{      22,381} & \multicolumn{1}{c}{      35,290} & \multicolumn{1}{c}{      27,532} \\ 
\hline\hline \multicolumn{6}{p{21cm}}{{\it Notes}: The table reports summary statistics on applicants' ROL. Panel A focuses on preference for funding. Specifically, Panel A shows shows the share of applicants who rank state-funded (self-funded) contracts exclusively, the share of students who rank state-funded and self-funded contracts as well, and the share of applicants who rank the same study program with state-funding and self-funding consecutively, and who rank the same study program with state-funding and self-funding consecutively on the top of the ROL. Panel B focuses on preference for program characteristics. Specifically, Panel B shows the share of applicants who rank exclusively contracts that are in the same settlement (single program location), at a single university (single university), and at a single faculty of a university (single university), and the share of applicants who rank exclusively contracts in a single field of study, and in a single major. Column (1) presents these shares for all applicants, and columns (2)--(5) report these shares for applicants residing in the capital, county towns, towns, and villages, respectively.}\end{tabular}

\end{table}
\end{landscape}

\clearpage
\begin{table}[htpb!]
\centering\footnotesize
	\captionsetup{justification=centering}
\caption{Socioeconomic status, academic achievement, and ROL characteristics}
\vspace{-1.5em}
\label{tab:ses_ROL}
\begin{tabular}{l c c c c c c c} \\ \hline\hline
Dependent variable & \multicolumn{3}{c}{Ranked state-funded contract} & & \multicolumn{3}{c}{Ranked self-funded contract} \\
                   & \multicolumn{3}{c}{exclusively}                  & & \multicolumn{3}{c}{exclusively} \\ \cline{2-4} \cline{6-8}
 & (1) & (2) & (3) & & (4) & (5) & (6) \\ \hline
Disadvantaged & 0.294$^{***}$ & & 0.277$^{***}$& & -0.172$^{***}$ & &  -0.152$^{***}$\\ 
 & (0.006) & &  (0.006) & & (0.002) & & (0.002)  \\ 
11th-grade GPA (standardized) &  &  0.079$^{***}$ & 0.079$^{***}$ & & & -0.032$^{***}$ & -0.032$^{***}$ \\ 
 &  & (0.002) & (0.002) &&&  (0.001) &  (0.001) \\ \hline
 Mean outcome (non-disadvd.) & 0.500 & 0.500 & 0.500 && 0.192 & 0.192 & 0.192 \\ 
R-squared                   & 0.015 & 0.035 & 0.049 && 0.009 & 0.047 & 0.054 \\ \hline\hline
\multicolumn{8}{p{16.5cm}}{{\it Notes}: The table presents the coefficient of linear regressions of ROL characteristics (such as whether an applicant ranked state-funded contracts exclusively on her ROL (columns (1)--(3)) and whether an applicant ranked self-funded contracts exclusively on her ROL (columns (4)--(6))) on socioeconomic status and academic achievement. The sample includes $     108,854$ applicants. The regressions include indicators for missing values of standardized 11th-grade GPA ($      23,043$ applicants). Robust standard errors are in parentheses.} \\ 
\multicolumn{8}{p{16.5cm}}{***: p<0.01, **: p<0.05, *: p<0.1.}
\end{tabular}

\end{table}

\clearpage
\begin{table}[htb!]
\centering\footnotesize
\captionsetup{justification=centering}
\caption{Socioeconomic status, academic achievement, and ROL characteristics (2)}
\label{app:tab:mp_ROL}
\vspace{-1.5em}
\begin{tabular}{l c c c c c c c} \\ \hline\hline
Dependent variable & \multicolumn{3}{c}{Same study program} & & \multicolumn{3}{c}{Same study program consecutively} \\
                   & \multicolumn{3}{c}{consecutively}                  & & \multicolumn{3}{c}{on the top of the ROL} \\ \cline{2-4} \cline{6-8}
 & (1) & (2) & (3) & & (4) & (5) & (6) \\ \hline
Disadvantaged & -0.075$^{***}$ & & -0.075$^{***}$& & -0.057$^{***}$ & &  -0.056$^{***}$\\ 
 & (0.004) & &  (0.004) & & (0.004) & & (0.004)  \\ 
11th-grade GPA (standardized) &  &  -0.034$^{***}$ & -0.034$^{***}$ & & & -0.027$^{***}$ & -0.027$^{***}$ \\ 
 &  & (0.001) & (0.001) &&&  (0.001) &  (0.001) \\ \hline
 Mean outcome (non-disadvd.) & 0.158 & 0.158 & 0.158 && 0.119 & 0.119 & 0.119 \\ 
R-squared                   & 0.002 & 0.007 & 0.009 && 0.001 & 0.006 & 0.007 \\ \hline\hline
\multicolumn{8}{p{16.5cm}}{{\it Notes}: The table presents the coefficient of linear regressions of ROL characteristics (such as whether an applicant ranked the same study program consecutively (columns (1)--(3)) and whether an applicant ranked the same study program consecutively on the top of the ROL (columns (4)--(6))) on socioeconomic status and academic achievement. The sample includes $     108,854$ applicants. The regressions include indicators for missing values of standardized 11th-grade GPA ($      23,043$ applicants). Robust standard errors are in parentheses.} \\ 
\multicolumn{8}{p{16.5cm}}{***: p<0.01, **: p<0.05, *: p<0.1.}
\end{tabular}

\end{table}

\clearpage
\begin{landscape}
\begin{table}[htpb!]
\centering\footnotesize
\captionsetup{justification=centering}
\caption{Socioeconomic status, academic achievement, and ROL characteristics: Additional specifications}
\label{tab:ses_ROL_robust}
\vspace{-1.5em}
\begin{tabular}{l c c c c c c c} \\ \hline\hline
Dependent variable & \multicolumn{3}{c}{Ranked state-funded contract} & & \multicolumn{3}{c}{Ranked self-funded contract} \\
                   & \multicolumn{3}{c}{exclusively}                  & & \multicolumn{3}{c}{exclusively} \\ \cline{2-4} \cline{6-8}
 & (1) & (2) & (3) & & (4) & (5) & (6) \\ \hline
Per-capita annual gross income (1000 USD, 2007 prices) & -0.028$^{***}$ & & & & 0.013$^{***}$ & & \\ 
 & (0.001) & &  & & (0.001)  & &  \\ 
NABC-based SES index & & -0.098$^{***}$ & & & & -0.010$^{***}$ & \\ 
 & & (0.005) & & & & (0.003) & \\ 
Capital & &  & -0.208$^{***}$ & & & &  0.078$^{***}$ \\ 
 & & & (0.002) & & & & (0.002) \\ 
County capital & &  & -0.061$^{***}$ & & & &  0.035$^{***}$ \\ 
 & & & (0.003) & & & & (0.003) \\ 
Town & &  & -0.043$^{***}$ & & & &  0.023$^{***}$ \\ 
 & & & (0.003) & & & & (0.003) \\ 
11th-grade GPA (standardized) &  0.081$^{***}$ &  0.077$^{***}$ & 0.081$^{***}$ & & -0.033$^{***}$ & -0.017$^{***}$ & -0.033$^{***}$ \\ 
 & (0.002) & (0.002) & (0.002) && (0.001) &  (0.001) &  (0.001) \\ \hline
R-squared                   & 0.057 & 0.088 & 0.056 && 0.055 & 0.145 & 0.053 \\  \hline\hline 
\multicolumn{8}{p{20.5cm}}{{\it Notes}: The table presents the coefficient of linear regressions of ROL characteristics (such as whether an applicant ranked state-funded contracts exclusively on her ROL (columns (1)--(3)) and whether an applicant ranked self-funded contracts exclusively on her ROL (columns (4)--(6))) on socioeconomic status and academic achievement. The sample includes $     108,854$ applicants. The regressions include indicators for missing values of standardized 11th-grade GPA ($      23,043$ applicants), per-capita annual gross income ($       1,920$ applicants), NABC-based SES index ($      24,399$ applicants), and residence ($       1,920$ applicants). The omitted category in columns (3) and (6) is village. Robust standard errors are in parentheses.} \\ 
\multicolumn{8}{p{16.5cm}}{***: p<0.01, **: p<0.05, *: p<0.1.}
\end{tabular}

\end{table}
\end{landscape}

\clearpage
\begin{landscape}
\begin{table}[htb!]
\centering\footnotesize
\captionsetup{justification=centering}
\caption{Socioeconomic status, academic achievement, and ROL characteristics: Additional specifications (2)}
\label{app:tab:mp_ROL_robust}
\vspace{-1.5em}
\begin{tabular}{l c c c c c c c} \\ \hline\hline
Dependent variable & \multicolumn{3}{c}{Same study program} & & \multicolumn{3}{c}{Same study program consecutively} \\
                   & \multicolumn{3}{c}{consecutively}                  & & \multicolumn{3}{c}{on the top of the ROL} \\ \cline{2-4} \cline{6-8}
 & (1) & (2) & (3) & & (4) & (5) & (6) \\ \hline
Per-capita annual gross income (1000 USD, 2007 prices) & 0.009$^{***}$ & & & & 0.009$^{***}$ & & \\ 
 & (0.000) & &  & & (0.000)  & &  \\ 
NABC-based SES index & & 0.061$^{***}$ & & & & 0.044$^{***}$ & \\ 
 & & (0.004) & & & & (0.003) & \\ 
Capital & &  & 0.073$^{***}$ & & & &  0.050$^{***}$ \\ 
 & & & (0.004) & & & & (0.004) \\ 
County capital & &  & 0.014$^{***}$ & & & &  0.014$^{***}$ \\ 
 & & & (0.003) & & & & (0.003) \\ 
Town & &  & 0.010$^{***}$ & & & &  0.007$^{***}$ \\ 
 & & & (0.002) & & & & (0.002) \\ 
11th-grade GPA (standardized) &  -0.035$^{***}$ &  -0.041$^{***}$ & -0.035$^{***}$ & & -0.028$^{***}$ & -0.031$^{***}$ & -0.028$^{***}$ \\ 
 & (0.001) & (0.001) & (0.001) && (0.001) &  (0.001) &  (0.001) \\ \hline
R-squared                   & 0.011 & 0.010 & 0.013 && 0.010 & 0.008 & 0.009 \\  \hline\hline 
\multicolumn{8}{p{20.5cm}}{{\it Notes}: The table presents the coefficient of linear regressions of ROL characteristics (such as whether an applicant ranked the same study program consecutively (columns (1)--(3)) and whether an applicant ranked the same study program consecutively on the top of the ROL (columns (4)--(6))) on socioeconomic status and academic achievement. The sample includes $     108,854$ applicants. The regressions include indicators for missing values of standardized 11th-grade GPA ($      23,043$ applicants), per-capita annual gross income ($       1,920$ applicants), NABC-based SES index ($      24,399$ applicants), and residence ($       1,920$ applicants). The omitted category in columns (3) and (6) is village. Robust standard errors are in parentheses.} \\ 
\multicolumn{8}{p{16.5cm}}{***: p<0.01, **: p<0.05, *: p<0.1.}
\end{tabular}

\end{table}
\end{landscape}

\clearpage
\begin{landscape}
\begin{table}[htb!]
\centering\footnotesize
\caption{Program characteristics by excess self-funded capacity at the benchmark}
\label{tab:excess_capacity}
\begin{tabular}{l c c c c c c c c} \\ \hline\hline
 & \multicolumn{2}{c}{Programs with no excess} & & \multicolumn{2}{c}{Programs with excess} & & \multicolumn{2}{c}{Programs with excess} \\ 
 & \multicolumn{2}{c}{self-funded capacity} & & \multicolumn{2}{c}{self-funded capacity} & & \multicolumn{2}{c}{self-funded capacity (weighted)} \\ \cline{2-3}\cline{5-6}\cline{8-9}
 & \multicolumn{1}{c}{(1)} & \multicolumn{1}{c}{(2)} & & \multicolumn{1}{c}{(3)} & \multicolumn{1}{c}{(4)} & & \multicolumn{1}{c}{(5)} & \multicolumn{1}{c}{(6)}  \\ \hline
Periphery & 0.59 & (0.49) & & 0.85 & (0.36) & & 0.84 & (0.37) \\ 
Full-time & 0.90 & (0.30) & & 0.72 & (0.45) & & 0.71 & (0.45) \\ 
Field: economics and business & 0.09 & (0.28) & & 0.10 & (0.30) & & 0.20 & (0.40) \\ 
Field: humanities & 0.15 & (0.36) & & 0.16 & (0.37) & & 0.11 & (0.32) \\ 
Field: engineering & 0.05 & (0.21) & & 0.09 & (0.29) & & 0.09 & (0.28) \\ 
Field: natural science & 0.05 & (0.21) & & 0.05 & (0.21) & & 0.02 & (0.15) \\ 
Field: social science & 0.08 & (0.27) & & 0.09 & (0.29) & & 0.09 & (0.29) \\ 
State-funded priority-score cutoff (scale: 0--144) & 104.65 & (21.17) & & 100.37 & (20.96) & & 102.39 & (20.97) \\ 
\hline Excess self-funded seats per program & 0.00 & (0.00) & & 29.53 & (33.27) & & 29.53 & (33.27) \\
\# number of programs &  494 & - & &  941 & - & &  941 & - \\ \hline\hline
\multicolumn{9}{p{21.5cm}}{{\it Notes}: The table presents means and standard deviations of various characteristics of programs with and without excess self-funded capacity at the benchmark. Columns (1)--(4) use equal weights, while columns (5)--(6) weight programs by the number of excess self-funded seats at the benchmark. Standard deviations are in parentheses.}\end{tabular}

\end{table}
\end{landscape}



\clearpage
\setcounter{table}{0}
\setcounter{figure}{0}
\renewcommand{\thetable}{D\arabic{table}}
\renewcommand{\thefigure}{D\arabic{figure}}

\clearpage
\section{Exercising Market Power More Broadly and More Narrowly}\label{app:narrow_and_broad}
This appendix expands our empirical analysis in two directions. First, we assess the consequences of being able to detect only a subset of students over which market power is applied by the preference flip algorithm. Second, we present alternate algorithms that apply market power more broadly and analyze their consequences.

\subsection{Exercising Market Power More Narrowly}
We consider the robustness of our results on the preference flip algorithm to cases where market power is applied more narrowly. Specifically, we apply market power only over a subset of applicants and rule out the possibility to apply market power over other applicants even if they rank two contracts with the same program consecutively. This analysis is motivated by the possibility that information frictions prevent the designer from identifying the entire set $A$ (from the description of the preference flip algorithm).

Appendix~\cref{app:tab:narrow_mp} presents the number of assigned applicants when market power is applied more narrowly, and the increase in the number of assigned applicants relative to the benchmark \SRDA assignment. We find that the increase in the number of assigned applicants is roughly linear in the fraction of applicants over which market power can be applied.

\begin{table}[htpb!]
\centering
\footnotesize
\caption{Narrow market power}
\label{app:tab:narrow_mp}
\begin{tabular}{c c c}\hline\hline
Fraction subject to market power (\%) & Assigned under preference flip & Increase relative to \SRDA (\%)\\
(1) & (2) & (3) \\ \hline
100 & 85,688 & 1.9 \\
90 & 85,538 & 1.7 \\ 
80 & 85,391 & 1.5 \\
70 & 85,236 & 1.3 \\
60 & 85,079 & 1.1 \\
50 & 84,929 & 0.9 \\
40 & 84,774 & 0.8 \\
30 & 84,609 & 0.6 \\
20 & 84,454 & 0.4 \\
10 & 84,293 & 0.2 \\ \hline\hline
\multicolumn{3}{p{16cm}}{{\it Notes}: The table presents the number of assigned applicants under the preference flip algorithm, when market power is applied only on a subset of applicants. In each row, we consider a fraction of applicants over which market power can be applied (from 10 to 100 percent), and rule out the possibility to apply market power over other applicants even if they rank two contracts with the same program consecutively. We present the average result over 50 random draws of subsets of this size. Column 3 presents the increase in the number of assigned students relative to the benchmark (84,130).}
\end{tabular}
\end{table}

\subsection{Exercising Market Power More Broadly}
In this section, we present two additional stable algorithms that can be interpreted as applying market power more broadly. We then repeat the analysis from the main text using these algorithms instead of the Preferences Flip Algorithm. 

\paragraph{Greedy Reject Algorithm} Initialize a set $B^{\prime}\subseteq B$, where $B$ is the set of student--college pairs such that the student is assigned to a state-funded seat in the program under \SPDA. For each pair in $B^{\prime}$, remove from the student's initial ROL the state-funded contract with the college (the student's assignment) and run \SPDA on the resulting problem. 
If the resulting allocation is certainly stable (with respect to original preferences), stop and output the resulting allocation. Otherwise, remove some pairs from $B^{\prime}$ and repeat the process. 

\paragraph{Combined Algorithm} Initialize sets $A^{\prime}\subseteq A$ and $B^{\prime}\subseteq B$. For each pair in $A^{\prime}$, flip the order of the contracts with the college in the student's original ROL, so that the state-funded contract appears immediately \textit{after} the self-funded contract. For each pair in $B^{\prime}$, remove from the student's initial ROL the state-funded contract with the college (the student's assignment). Run \SPDA on the resulting problem. 
If the resulting allocation is certainly stable (with respect to original preferences), stop and output the resulting allocation. Otherwise, remove some pairs from $A^{\prime}$ or from $B^{\prime}$ and repeat the process.

\begin{proposition} \label{prop:algo_appendix}
For any initial set(s)  $B'\subseteq B$ (and $A'\subseteq A$) and any rule regarding the choice of elements to remove from $B'$ (and $A'$), the greedy reject (combined) algorithm results in a stable allocation. Furthermore, the resulting allocation may not allocate funding based on merit. 
\end{proposition}
\begin{proof}
The proof is identical to the proof of \cref{prop:algo} and therefore omitted. 
\end{proof}

We now repeat the analysis from the main text using the greedy reject algorithm and the combined algorithm instead of the preference flip algorithm.
Appendix Table~\ref{app:tab:core_size} compares the outcomes of the greedy reject and the combined algorithms to the \SRDA benchmark. Appendix Tables~\ref{tab:student_comparison_h1} and \ref{tab:student_comparison_h3} compare the characteristics of winners to losers under the greedy reject algorithm and the combined algorithm, respectively. Appendix Tables~\ref{tab:mobility1} and \ref{tab:mobility3} investigate geographic mobility under the greedy reject algorithm and the combined algorithm, respectively. The results are qualitatively similar to those from the main text, but their magnitudes are slightly larger. 

\begin{table}[h!]
\centering\footnotesize
\captionsetup{justification=centering}
\caption{Winners, losers, and the number of assigned applicants: Greedy reject and combined algorithms}
\vspace{-1.5em}
\label{app:tab:core_size}
\begin{tabular}{l S[table-number-alignment = center, table-format = 5.0] S[table-number-alignment = center, table-format = 5.0] S[table-number-alignment = center, table-format = 5.0] S[table-number-alignment = center, table-format = 5.0]} \\ \hline\hline
 & \multicolumn{1}{c}{Benchmark} & & \multicolumn{1}{c}{Greedy reject} & \multicolumn{1}{c}{Combined} \\ 
 & \multicolumn{1}{c}{SR-DA} & \multicolumn{1}{c}{SP-DA} & \multicolumn{1}{c}{algorithm} & \multicolumn{1}{c}{algorithm}  \\ 
 & \multicolumn{1}{c}{(1)} & \multicolumn{1}{c}{(2)} & \multicolumn{1}{c}{(3)} & \multicolumn{1}{c}{(4)} \\ \hline
\multicolumn{5}{l}{A. {\it Number of assigned applicants}}\\ 
Assigned to a contract &        84130 &        84130 &         85575 &        85847 \\ 
Assigned to a state-funded contract &        48725 &        48725 &        48710 &                48706 \\ 
Assigned to a self-funded contract &        35405 &        35405 &        36865 &               37141 \\ 
Assigned to a self-funded, full-time contract &        13321 &        13321 &        14410 &                14565  \\ 
\multicolumn{5}{l}{}\\ 
B. {\it Winners} & &            8&         5056 &          5283  \\ 
 -- Newly assigned &  &            0&         1969 &          2346  \\ 
 -- New program & &            8 &         2168 &          2498  \\ 
 -- Same program, preferred financial terms & &            0 &          919 &                    439  \\ 
\multicolumn{5}{l}{}\\ 
C. {\it Losers} &  &            0 &         4570 &                  4734   \\ 
 -- Newly unassigned & &            0 &          524 &           629  \\ 
 -- New program & &            0 &          645 &                    762  \\ 
 -- Same program, less preferred financial terms &   &            0 &         3401 &                  3343 \\ \hline\hline
\multicolumn{5}{p{15.5cm}}{{\it Notes}: The table presents the number of assigned applicants under various stable assignments (Panel A). Panels B and C describe the number of winners and losers from changing the benchmark \SRDA to \SPDA and to our alternate algorithms. An applicant is a winner if she is assigned to a contract she ranked higher than her contract in the benchmark assignment. An applicant is a loser if she is assigned to a contract she ranked lower than her assigned contract in the benchmark assignment or if she becomes unassigned. Column (1) shows the benchmark assignment. Column (2) presents the \SPDA assignment. Column (3) presents the greedy reject algorithm assignment and column (4) presents the combined algorithm assignment.} \\ 
\end{tabular}

\end{table}

\clearpage
\begin{table}[htpb!]
\centering\footnotesize
\captionsetup{justification=centering}
\caption{Characteristics of winners and losers: Greedy reject algorithm}
\label{tab:student_comparison_h1}
\vspace{-1.5em}
\begin{tabular}{l S[table-number-alignment = center-decimal-marker] S[table-number-alignment = center-decimal-marker] S[table-number-alignment = center-decimal-marker]} \\ \hline\hline
 & \multicolumn{1}{c}{Winners} & \multicolumn{1}{c}{Losers} & \multicolumn{1}{c}{p-values: (1)=(2)} \\ 
 & \multicolumn{1}{c}{(1)} & \multicolumn{1}{c}{(2)} & \multicolumn{1}{c}{(3)}  \\ \hline
 Disadvantaged & 0.057 &  0.026 & 0.000 \\ 
 Per-capita annual gross income (1000 USD, 2007 prices) & 9.735 &  10.385 & 0.000 \\ 
 NABC-based SES index & 0.180 &  0.306 & 0.000 \\ 
 Capital & 0.188 &  0.277 & 0.000 \\ 
 County capital & 0.192 &  0.217 & 0.003 \\ 
 Town & 0.329 &  0.298 & 0.001 \\ 
 Village & 0.291 &  0.208 & 0.000 \\ 
 11th-grade GPA (1--5) & 3.473 &  3.704 & 0.000 \\ 
 Female & 0.562 &  0.553 & 0.355 \\ 
\hline Number of applicants & \multicolumn{1}{c}{       5,056} & \multicolumn{1}{c}{       4,570} & \multicolumn{1}{c}{       9,626} \\ \hline\hline 
\multicolumn{4}{p{15.5cm}}{{\it Notes}: The table reports the mean values of characteristics of winners and losers from changing the benchmark \SRDA to the greedy reject algorithm. An applicant is a winner if she is assigned to a contract she ranked higher than her benchmark assignment. An applicant is a loser if she is assigned to a contract she ranked lower than her assigned contract in the benchmark assignment or if she becomes unassigned. Column (3) presents p-values for the equality of mean characteristics of winners and losers.}\end{tabular}

\end{table}

\clearpage
\begin{table}[htpb!]
\centering\footnotesize
\captionsetup{justification=centering}
\caption{Characteristics of winners and losers: Combined algorithm}
\label{tab:student_comparison_h3}
\vspace{-1.5em}
\begin{tabular}{l S[table-number-alignment = center-decimal-marker] S[table-number-alignment = center-decimal-marker] S[table-number-alignment = center-decimal-marker]} \\ \hline\hline
 & \multicolumn{1}{c}{Winners} & \multicolumn{1}{c}{Losers} & \multicolumn{1}{c}{p-values: (1)=(2)} \\ 
 & \multicolumn{1}{c}{(1)} & \multicolumn{1}{c}{(2)} & \multicolumn{1}{c}{(3)}  \\ \hline
 Disadvantaged & 0.061 &  0.026 & 0.000 \\ 
 Per-capita annual gross income (1000 USD, 2007 prices) & 9.718 &  10.413 & 0.000 \\ 
 NABC-based SES index & 0.177 &  0.305 & 0.000 \\ 
 Capital & 0.183 &  0.282 & 0.000 \\ 
 County capital & 0.191 &  0.214 & 0.005 \\ 
 Town & 0.331 &  0.297 & 0.000 \\ 
 Village & 0.295 &  0.207 & 0.000 \\ 
 11th-grade GPA (1--5) & 3.483 &  3.689 & 0.000 \\ 
 Female & 0.570 &  0.556 & 0.160 \\ 
\hline Number of applicants & \multicolumn{1}{c}{       5,283} & \multicolumn{1}{c}{       4,734} & \multicolumn{1}{c}{      10,017} \\ \hline\hline 
\multicolumn{4}{p{15.5cm}}{{\it Notes}: The table reports the mean values of characteristics of winners and losers from changing the benchmark \SRDA to the combined algorithm. An applicant is a winner if she is assigned to a contract she ranked higher than her benchmark assignment. An applicant is a loser if she is assigned to a contract she ranked lower than her assigned contract in the benchmark assignment or if she becomes unassigned. Column (3) presents p-values for the equality of mean characteristics of winners and losers.}
\end{tabular}

\end{table}

\clearpage

\begin{table}[htpb!]
\captionsetup{justification=centering}
\centering\footnotesize
\caption{Geographic mobility: Greedy reject algorithm}
\label{tab:mobility1}
\vspace{-1.5em}
\begin{tabular}{l c c c c c c c c} \\ \hline\hline
 & \multicolumn{2}{c}{Mover} & & \multicolumn{2}{c}{Assigned to periphery} & & \multicolumn{2}{c}{Moved to capital} \\ \cline{2-3} \cline{5-6} \cline{8-9}
& \multicolumn{1}{c}{\SRDA} & \multicolumn{1}{c}{GRA} & & \multicolumn{1}{c}{\SRDA} & \multicolumn{1}{c}{GRA} & & \multicolumn{1}{c}{\SRDA} & \multicolumn{1}{c}{GRA} \\ 
 & \multicolumn{1}{c}{(1)} & \multicolumn{1}{c}{(2)} & & \multicolumn{1}{c}{(3)} & \multicolumn{1}{c}{(4)} & & \multicolumn{1}{c}{(5)} & \multicolumn{1}{c}{(6)}\\ \hline
\multicolumn{9}{l}{A. {\it Total}} \\ 
All assigned applicants &       51,221 &       52,249 & &       47,845 &       48,876 & &       22,444 &       22,804 \\ 
\multicolumn{9}{l}{} \\ 
\multicolumn{9}{l}{B. {\it Winners}} \\
Newly assigned (N=       1,969) & -- &        1,249 & & -- &        1,362 & & -- &          402 \\
New program (N=       2,168) &        1,379 &        1,415 & &        1,323 &        1,147 & &          521 &          660 \\
\multicolumn{9}{l}{} \\
\multicolumn{9}{l}{C. {\it Losers}} \\
Newly unassigned (N=         524) &          279 & -- & &          226 & -- & &          143 & -- \\
New program (N=         645) &          348 &          370 & &          256 &          327 & &          191 &          153 \\ \hline\hline 
\multicolumn{9}{p{16cm}}{{\it Notes}: The table presents the number of movers (i.e., applicants whose assigned study program is not located in the county where they reside), of applicants assigned to the periphery (i.e., not to the capital), and of applicants who moved to the capital (i.e., applicants who get assigned to a study program in the capital but do not reside there) under the greedy reject algorithm (GRA) and in the benchmark \SRDA assignments. Panel A presents the total number of applicants with these characteristics. Panels B and C focus on winners and losers who are assigned to a different study program.}\end{tabular}

\end{table}

\clearpage

\begin{table}[htpb!]
\captionsetup{justification=centering}
\centering\footnotesize
\caption{Geographic mobility: Combined algorithm}
\label{tab:mobility3}
\vspace{-1.5em}
\begin{tabular}{l c c c c c c c c} \\ \hline\hline
 & \multicolumn{2}{c}{Mover} & & \multicolumn{2}{c}{Assigned to periphery} & & \multicolumn{2}{c}{Moved to capital} \\ \cline{2-3} \cline{5-6} \cline{8-9}
& \multicolumn{1}{c}{\SRDA} & \multicolumn{1}{c}{CA} & & \multicolumn{1}{c}{\SRDA} & \multicolumn{1}{c}{CA} & & \multicolumn{1}{c}{\SRDA} & \multicolumn{1}{c}{CA} \\ 
 & \multicolumn{1}{c}{(1)} & \multicolumn{1}{c}{(2)} & & \multicolumn{1}{c}{(3)} & \multicolumn{1}{c}{(4)} & & \multicolumn{1}{c}{(5)} & \multicolumn{1}{c}{(6)}\\ \hline
\multicolumn{9}{l}{A. {\it Total}} \\ 
All assigned applicants &       51,221 &       52,415 & &       47,845 &       49,073 & &       22,444 &       22,851 \\ 
\multicolumn{9}{l}{} \\ 
\multicolumn{9}{l}{B. {\it Winners}} \\
Newly assigned (N=       2,346) & -- &        1,463 & & -- &        1,620 & & -- &          472 \\
New program (N=       2,498) &        1,582 &        1,624 & &        1,526 &        1,316 & &          602 &          769 \\
\multicolumn{9}{l}{} \\
\multicolumn{9}{l}{C. {\it Losers}} \\
Newly unassigned (N=         629) &          341 & -- & &          265 & -- & &          183 & -- \\
New program (N=         762) &          419 &          449 & &          294 &          377 & &          232 &          183 \\ \hline\hline 
\multicolumn{9}{p{16cm}}{{\it Notes}: The table presents the number of movers (i.e., applicants whose assigned study program is not located in the county where they reside), of applicants assigned to the periphery (i.e., not to the capital), and of applicants who moved to the capital (i.e., applicants who get assigned to a study program in the capital but do not reside there) under the combined algorithm and in the benchmark \SRDA assignments. Panel A presents the total number of applicants with these characteristics. Panels B and C focus on winners and losers who are assigned to a different study program.}\end{tabular}

\end{table}

\clearpage
\setcounter{table}{0}
\setcounter{figure}{0}
\renewcommand{\thetable}{E\arabic{table}}
\renewcommand{\thefigure}{E\arabic{figure}}

\section{Robustness Analysis (for Online Publication)}\label{app:section:empirical_findings}
This Appendix shows that our empirical findings are robust to the way we handle minor inconsistencies in our data. Our data contain 3,686 applicants (3.4 percent of all applicants) whose reported assignment and reported priority scores are inconsistent with stability (their reported priority score does not exceed the priority-score cutoff of their reported assignment, or there is an alternative in their ROL from which they are rejected even though their reported priority score exceeds the priority-score cutoff of this alternative). These inconsistencies largely stem from applicants having a reported priority score of zero. Our main approach, which is presented in the text, holds the assignment of applicants with such inconsistencies fixed to their reported assignment, and makes the corresponding seats unavailable to others. 

In this appendix, we consider two alternative approaches to verify that this issue does not drive our empirical findings:
\begin{itemize}
	\item Approach 1: We hold the assignment of applicants with inconsistencies fixed to their reported assignment, but we keep the corresponding seats available to others (i.e., we do not reduce the corresponding quotas)
	\item Approach 2: We do not hold the assignment of applicants with inconsistencies fixed, and we keep the corresponding seats available to others.
\end{itemize}

Approach~1 changes the \SRDA (benchmark) assignment of 4,439 applicants and increases the number of assigned applicants in the benchmark by 1,989 relative to our main approach. Approach~2 changes the \SRDA (benchmark) assignment of 6,052 applicants and increases the number of assigned applicants in the benchmark by 162 relative to our main approach. In spite of these changes, \Cref{app:fig:winners_and_losers_robust} shows that our main findings continue to hold.


Panel~A of \cref{app:fig:winners_and_losers_robust} presents the difference between the number of assigned applicants under the preference flip algorithm and under \SRDA. In the main text, this corresponds to the difference between column (3) and column (1) in Panel~A of \cref{tab:core_size}. Panel~B presents the number of winners from changing the benchmark to the preference flip algorithm. In the main text, this corresponds to column (3) in Panel~B of \cref{tab:core_size}. Panel~C presents the number of losers from changing the benchmark to the preference flip algorithm. In the main text, this corresponds to column (3) in Panel~C of \cref{tab:core_size}. Since the numbers in each of the panels are almost identical, we conclude that our main findings continue to hold.


\clearpage
\begin{figure}[htpb!]
	\centering
	\captionsetup{justification=centering}
	\caption{Winners, losers, and the number of assigned applicants: Robustness}
	\label{app:fig:winners_and_losers_robust}
	\includegraphics[scale = 0.8]{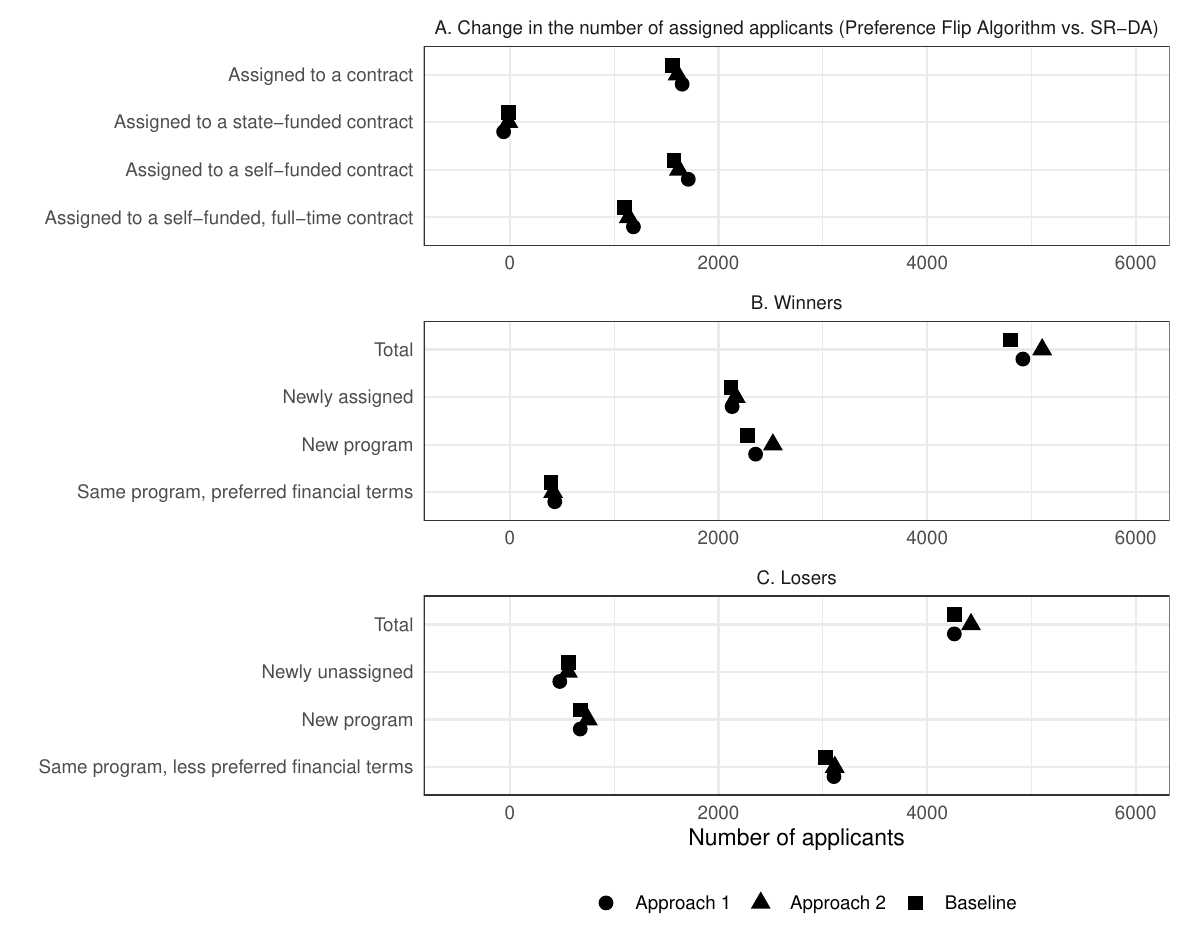}
	\captionsetup{justification=justified}
	\caption*{{\small {\it Notes}: This figure studies the robustness of our empirical findings to the way we handle minor inconsistencies in our data. Panel~A presents the difference between the number of assigned applicants under the preference flip algorithm and under \SRDA. In the main text (baseline), this corresponds to the difference between column (3) and column (1) in Panel~A of \cref{tab:core_size}. Panel~B presents the number of winners from changing the benchmark to the preference flip algorithm. In the main text, this corresponds to column (3) in Panel~B of \cref{tab:core_size}. Panel~C presents the number of losers from changing the benchmark to the preference flip algorithm. In the main text, this corresponds to column (3) in Panel~C of \cref{tab:core_size}. The baseline corresponds to the analysis in the main text.}}
\end{figure}



\clearpage
\section{Variable Description (for Online Publication)}
This appendix describes the construction of variables in\cref{app:tab:summary_stat_applicants}.

\begin{itemize}
	\item {\bf Disadvantaged status}: Applicants receive priority points for having a low socioeconomic background. Our administrative data report these priority points. If an applicant received priority points for this reason in any of the alternatives in her ROL, we label this applicant as disadvantaged. Source: Administrative data on college admissions.
	\item {\bf NABC-based SES index}: First, we compute the NABC-based SES index for each survey respondent between 2008 and 2012 (5 years) in grade 10. We then compute the high-school-specific means over the 2008--2012 period. Finally, we merge the high-school-specific NABC-based SES index with the applicants using their high-school identifier. The high-school identifier is missing for 15,521 applicants. Source: National Assessment of Basic Competencies (NABC).
	\item {\bf Per-capita annual gross income}: Settlement-level income per population. Applicants come from 2,804 different settlements. Average exchange rate in 2007: 183HUF/USD. Source: T-STAR dataset (\url{http://adatbank.krtk.mta.hu/adatbazisok___tstar})
	\item {\bf Capital, county capital, town, village}: Dummy variables for the type of settlement where the applicants reside. Source: Administrative data on college admissions.
	\item {\bf 11th-grade GPA}: Average grade in mathematics, Hungarian grammar and literature, and history in grade 11. Source: Administrative data on college admissions.
	\item {\bf Female}: Dummy variable for being female. Source: Administrative data on college admissions.
	\item {\bf Number of alternatives in ROL}: Our administrative data include information on the number of alternatives in each applicant's ROL. Source: Administrative data on college admissions.
	\item {\bf Number of programs in ROL (observed)}: Our administrative data report the first 6 alternatives in an applicant's ROL as well as the applicant's realized assignment, in case it was ranked lower. We compute the number of programs in an applicant's ROL based on this information. Source: Administrative data on college admissions.
\end{itemize}

\clearpage
\section{Evidence from the Israeli Psychology Master's Match (for Online Publication)}\label{app:ipmm}
In this Appendix, we evaluate our theoretical predictions using data from the Israeli Psychology Master's Match (IPMM). In this market, unlike in Hungary, applicants are not heterogeneous in the sensitivity to financial terms. Consistent with our theory, we find that the preference flip algorithm does not increase the number of applicants admitted to college relative to the outcome of \SPDA. 

\subsection{Background and Data}
We use data from the 2014 and 2015 rounds of the IPMM. In each of these years, there are approximately 1,000 applicants, of which approximately 600 are assigned. 

Several departments offer a limited number of no-strings-attached scholarships to selected students. In 2014, 10 programs in 3 departments offered admission under multiple financial terms. The number increased to 15 programs in 4 departments in 2015. Funding levels ranged from approximately \$2,000 a year to approximately \$20,000 a year. The number of available scholarships was 25 in 2014 and 36 in 2015. The centralized clearinghouse uses a version of \SPDA which allows programs to offer contracts with multiple financial terms, and allows applicants to rank these alternatives separately.\footnote{While colleges' preferences in the IPMM are more complex than in Hungarian college admissions \citep{hrspp}, they meets the hidden substitutes condition of \citet{hk2015} and so \SPDA is stable and strategyproof.} 

\subsection{The Distribution of Applicants' Rank-Order Lists}
Applicants' ROLs include 4.32 contracts on average (SD=4.14).  About 37.2\% of
the ROLs include at least one contract with a program that offered admission
under multiple financial terms. Of these, only 3.4\% include only the
funded contract in some program but not the unfunded contract (cf. \Cref{tab:summary_stat_rols}). 

Among the ROLs that include both a funded and an unfunded contract with the same program, more than 90\% ranked the funded contract first. Among these ROLs, 82.6\% ranked both contracts consecutively, and the mean number of contracts ranked between a funded contract and the corresponding unfunded contract was 0.34. 

In summary, an overwhelming majority of applicants in the IPMM are not sensitive to financial terms.

\subsection{The Preference Flip Allocation}
We next compare the stable allocations resulting from \SPDA and from the preference flip algorithm.\footnote{The proof that the preference flip algorithm results in a stable allocation continues to hold in the more general preference structure of the IPMM (since \SPDA continues to be stable).} 
We implement the preference flip algorithm initializing $A' = A$ and with multiple rules of removal: in addition to the rule used in the main analysis, we also remove the lowest-ranked student (instead of the highest-ranked student), and furthermore we use 100 randomly generated rules of removal.
Consistent with \cref{prop:insens}, none of these implementations result in an increase in the number of assigned students, and the largest number of affected students in 4.

\end{appendices}
\end{document}